\documentclass[10pt]{article}
\usepackage[T1]{fontenc}
\usepackage{graphicx}
\title{Approximating the Earth Mover's Distance between sets of geometric objects\thanks{Supported by the Netherlands Organisation for Scientific Research (NWO) under project no. 612.001.651.}}
\author{Marc van Kreveld \and
Frank Staals \and
Amir Vaxman \and Jordi Vermeulen}
\date{}

\usepackage{amsmath,amssymb}
\usepackage{amsthm}
\usepackage{mathtools}
\usepackage{booktabs}
\usepackage{url}
\usepackage[labelformat=simple]{subcaption}
\usepackage[hidelinks]{hyperref}
\usepackage[capitalize,noabbrev]{cleveref}
\usepackage[margin=3cm]{geometry}
\usepackage{lineno}

\graphicspath{{./figures/}}

\newtheorem{theorem}{Theorem}
\newtheorem{lemma}[theorem]{Lemma}
\newtheorem{corollary}[theorem]{Corollary}

\DeclareMathOperator*{\argmin}{arg\,min}
\newcommand{\dd}{\mathop{}\,\mathrm{d}}         
\newcommand{\demd}{d_e}                         
\newcommand{\eps}{\varepsilon}

\newcommand{\norm}[1]{\left|#1\right|}          

\DeclareMathOperator{\polylog}{polylog}
\newcommand{\R}{\mathbb{R}}                     
\newcommand{\tp}{\eta}                          
\newcommand{\tpf}{H}                            

\let\originalleft\left
\let\originalright\right
\renewcommand{\left}{\mathopen{}\mathclose\bgroup\originalleft}
\renewcommand{\right}{\aftergroup\egroup\originalright}

\begin{document}

\maketitle

\begin{abstract}
    Given two distributions \(P\) and \(S\) of equal total mass, the Earth Mover's Distance measures the cost of transforming one distribution into the other, where the cost of moving a unit of mass is equal to the distance over which it is moved.

    We give approximation algorithms for the Earth Mover's Distance between various sets of geometric objects.
    We give a \((1 + \eps)\)-approximation when \(P\) is a set of weighted points and \(S\) is a set of line segments, triangles or \(d\)-dimensional simplices.
    When \(P\) and \(S\) are both sets of line segments, sets of triangles or sets of simplices, we give a \((1 + \eps)\)-approximation with a small additive term.
    All algorithms run in time polynomial in the size of \(P\) and \(S\), and actually calculate the transport plan (that is, a specification of how to move the mass), rather than just the cost.
    To our knowledge, these are the first combinatorial algorithms with a provable approximation ratio for the Earth Mover's Distance when the objects are continuous rather than discrete points.
\end{abstract}

\section{Introduction}
The earth mover's distance (EMD) is a metric that is widely used in fields such as image retrieval~\cite{rubner2000}, shape matching~\cite{grauman2004,memoli2009,su2015} and mesh reconstruction~\cite{degoes2011}.
It models two sets \(P\) and \(S\) as distributions of mass, and takes their distance \(\demd(P,S)\) to be the minimum cost of transforming one distribution into the other, where cost is measured by the amount of mass moved multiplied by the distance over which it is moved.
More formally,
\begin{linenomath*}
\begin{equation*}
    \demd(P,S) = \inf_{\tp \in \tpf}{\int_{P}{\int_{S}{d(p,s) \cdot \tp(p,s)\dd p \dd s}}}
\end{equation*}
\end{linenomath*}

\noindent
where \(\tpf\) is the set of all mappings of mass between \(P\) and \(S\) and \(d(\cdot, \cdot)\) is any metric.
In the case where \(P\) and \(S\) are finite sets of (weighted) points, we can rewrite this as
\begin{linenomath*}
\begin{equation*}
    \demd(P,S) = \min_{\tp \in \tpf}{\sum_{p \in P}{\sum_{s \in S}{d(p,s)} \cdot \tp(p,s)}}
\end{equation*}
\end{linenomath*}

\noindent
For unweighted point sets, the solution can be obtained by solving an assignment problem; for weighted point sets, this is an instance of a minimum cost flow problem.

Recently, much attention has been devoted to computing the earth mover's distance when both \(P\) and \(S\) are sets of points~\cite{agarwal2017,fox2022,khesin2021,sharathkumar2012}.
In this paper we expand on this by letting \(P\) and \(S\) be sets of points, line segments, triangles or \(d\)-dimensional simplices in \(\mathbb{R}^d\).
We describe a unified framework for calculating the EMD between points and segments, points and triangles, points and simplices, segments and segments, triangles and triangles and simplices and simplices.
Our approach provides polynomial-time algorithms that give a \((1 + \eps)\)-approximation to the earth mover's distance between \(P\) and \(S\), for some arbitrarily small \(\eps > 0\).
Moreover, our algorithms produce an assignment of mass that realises this cost.
For triangles and simplices, the running time also depends on the largest edge length (note that we normalise the total area/volume of each set to one, so we cannot improve the running time by scaling the input).
When neither set consists of discrete points, there is a small extra additive term in our approximation.
For all our algorithms, our approach is to subdivide the elements of the input into sufficiently small pieces, and then approximate each piece by a point.
The approximate optimal transport plan can then be obtained by solving a transport problem on these points.
Our results are summarised in \cref{tab:results}.
Note that all our algorithms give the solution with high probability; this is simply a consequence of using Fox and Lu's algorithm~\cite{fox2022}~to solve the optimal transport problem on points.
Substituting a deterministic algorithm here would make our results deterministic as well.

To our knowledge, these are the first combinatorial algorithms with a provable approximation ratio for the earth mover's distance when the objects are continuous rather than discrete points.
We give algorithms for moving mass from points to segments (\cref{sec:point-segment}), points to triangles (\cref{sec:point-triangle}), points to simplices (\cref{sec:point-simplex}), segments to segments (\cref{sec:segment-segment}), triangles to triangles (\cref{sec:triangle-triangle}) and simplices to simplices (\cref{sec:simplex-simplex}).

\begin{table}[!b]
    \centering
    \begin{tabular}{lll}
        \toprule
        Objects                 & Running time & Additive term\\
        \midrule
        Points to segments      & \(O\left(\frac{nm}{\eps^c}\polylog{\frac{nm}{\eps}}\right)\) & - \\
        Points to triangles     & \(O\left(\frac{nm}{\eps^c}\polylog{\frac{nm\Delta}{\eps}}\right)\) & - \\
        Points to simplices     & \(O\left(6^dd^2d^dm + \frac{105^dd^2d^{d/2}nm}{\eps^{O(d)}} \log^{O(d)}\left(\frac{dnm\Delta}{\eps^d}\right)\right)\) & - \\
        Segments to segments    & \(O\left(\frac{nm}{\eps^c}\polylog{\frac{nm}{\eps}}\right)\) & \(O\left(\frac{\eps}{nm}\right)\) \\
        Triangles to triangles  & \(O\left(\frac{nm\Delta(n + m)}{\eps^c}\polylog{\frac{nm\Delta}{\eps}}\right)\) & \(O\left(\frac{\eps}{\sqrt{nm}}\right)\) \\
        Simplices to simplices  & \(O\left(\frac{\sqrt{d}(nm)^{1/d}\Delta^d(n+m)}{\eps^{O(d)}}\log^{O(d)}\left(\frac{d(nm)^{1/d}\Delta^d}{\eps}\right)\right)\) & \(O\left(\frac{\sqrt{d}\eps}{(nm)^{1/d}}\right)\) \\
        \bottomrule \\
    \end{tabular}
    \caption{A summary of our results for different choices of sets \(P\) and \(S\) of sizes \(n\) and \(m\).
    \(d\) is the dimension, \(\Delta\) is the largest diameter of any element of the sets.}
    \label{tab:results}
\end{table}

\section{Related work}\label{sec:emd-related-work}
The general problem of optimally moving a distribution of mass was first described by Monge in 1781~\cite{monge1781}, and was reformulated by Kantorovich in 1942~\cite{kantorovich1942}.
It is known as the earth mover's distance due to the analogy of moving piles of dirt around; it is also known as the 1-Wasserstein distance, and is a special case of the more general optimal transport problem.
For a full treatment of the problem's history and connections to other areas of mathematics, the reader is referred to Villani's book~\cite{villani2008}.

The earth mover's distance has been studied in many geometric contexts.
Agarwal et al.~\cite{agarwal2017}~give both exact and approximation algorithms for the case where both sets are points under some \(L_p\) metric.
When both sets are weighted points, Khesin at al.~\cite{khesin2021}~give two algorithms running in \(O(n \eps^{-O(d)} \log(\Lambda)^{O(d)} \log n \log(1/\rho))\) and \(O(n \eps^{-O(d)} \log{(U}^{O(d)} \log(n/\rho)^2)\) time that compute a \((1 + \eps)\)-approximation with probability at least \(1 - \rho\), where \(d\) is the dimension, \(\Lambda\) the aspect ratio of the input, and \(U\) the total mass.
However, their algorithm assumes that the point weights are integers, whereas our weights can be arbitrary real numbers, as they correspond to lengths and areas.
This result was improved by Fox and Lu~\cite{fox2022}. They used a similar method to obtain, with high probability, a \((1 + \eps)\)-approximation in \(O(n \eps^{-O(d)}\log^{O(d)} n)\) time.
The EMD was also studied when the input sets may be transformed: Cabello et al.~\cite{cabello2008} present algorithms that, given two weighted point sets of \(n\) and \(m\) points in \(\R^2\), compute a \((1 + \eps)\)-approximation of a translation that minimises the EMD, and a \((2 + \eps)\)-approximation of a rigid motion that minimises the EMD.
These algorithms run in \(O((n^2/\eps^4)\log^2 n)\) and \(O((n^3m/\eps^4)\log^2 n)\) time, respectively.

For continuous distributions, rather than discrete point sets, many numerical algorithms are known (see e.g.\ De Goes et al.~\cite{degoes2012}, Lavenant et al.~\cite{lavenant2018}, M\'erigot~\cite{merigot2011,merigot2018} and Solomon et al.~\cite{solomon2015}).
For the case where one set consists of weighted points and the other is a bounded set \(C \subset \mathbb{R}^d\), Gei\ss{} et al.~\cite{geiss2013} give a geometric proof that there exists an additively weighted Voronoi diagram such that transporting mass from each point \(p\) to the part of \(C\) contained in its Voronoi cell is optimal.
The weights of this Voronoi diagram can be determined numerically.

De Goes et al.~\cite{degoes2011} discuss a problem similar to our own, but in the context of the reconstruction and simplification of 2D shapes.
Given a set of points, they want to reconstruct a simplicial complex of a given number of vertices that closely represents the shape of the point set.
They start with computing the Delaunay triangulation of the point set, then iteratively collapse the edge that minimises the increase in the EMD between the point set and the triangulation.
They use a variant of the EMD in which the cost is proportional to the square of the distance (2-Wasserstein distance).
This allows them to calculate this variant of the EMD between a given set of points and a given edge of the triangulation exactly, as the squared distance can be decomposed into a normal and a tangential component.
However, they determine the assignment of points to edges heuristically.
In this work, we show how to obtain a \((1 + \eps)\)-approximation to the true optimal solution.

\section{Preliminaries}\label{sec:emd-preliminaries}
We are given a set of points \(P = \{p_1, \dots, p_n\}\) in the plane with weights \(\norm{p_i}\) and a set of geometric objects \(S = \{s_1, \dots, s_m\}\), with lengths, areas or volumes \(\norm{s_j}\).
It is given that \(\sum{\norm{p_i}} = \sum{\norm{s_i}}\).
We assume the mass associated with an object \(s_i\) is distributed uniformly over the object, and that all objects have the same mass density.
For convenience, and without loss of generality, we scale the input such that the total mass in either set is one.
We want to compute a ``transport plan'' of mass from \(P\) to \(S\) that minimises the cost according to the earth mover's distance.
We define for each pair \((p_i, s_j) \in P \times S\) a function \(\tp_{i,j}(x,y)\), that describes the density of mass being moved from \(p_i\) to the point \((x,y) \in s_j\).
All these functions together describe the function \(\tp\) used in the definition of \(\demd(A,B)\).
Such a set of functions needs to satisfy the following conditions to be a valid transport plan:
\begin{linenomath*}
\begin{align*}
    \forall i, j:           & \quad 0 \leq \tp_{i,j}(x,y) \leq 1                                \\
    \forall i:              & \quad \sum_{j=1}^{m}\int_{s_j}{\tp_{i,j}(x,y)\dd t} = \norm{p_i}  \\
    \forall j,(x,y)\in s_j: & \quad \sum_{i=1}^{n}{\tp_{i,j}(x,y)} = 1
\end{align*}
\end{linenomath*}

\noindent
We can then define the cost $\norm{\tp}$ of a given transport plan \(\tp\) as
\begin{linenomath*}
\begin{equation*}
    \norm{\tp} = \sum_{i=1}^{n}{\sum_{j=1}^{m}{\int_{s_j}{\tp_{i, j}(x,y) \cdot d(p_i, (x,y))\dd t}}}
\end{equation*}
\end{linenomath*}

\noindent
where \(d(\cdot,\cdot)\) is any metric.
Our problem is to find a transport plan \(\tp^*\) with minimal cost.

In the following section, we give an exact algorithm to calculate an optimal transport plan between a set of weighted points and line segments when \(d(\cdot,\cdot)\) is the \(L_1\) metric.
However, the approach we use does not seem to generalise to Euclidean distances, objects with areas, or even two sets of segments.
This motivated us to look towards approximation algorithms for more general versions of the EMD problem.
In the rest of this paper, we describe approximation algorithms and only consider the case where \(d(\cdot,\cdot)\) is the \(L_2\) metric.

\section{Points to segments under the \texorpdfstring{\(L_1\)}{L1} metric}
When \(S\) is a set of line segments and distances are measured by the \(L_1\) metric, we can solve the problem exactly by a convex quadratic program.
We first subdivide all segments on the \(x\)- and \(y\)-coordinates of the points; call the set of subdivided segments \(Q = \{q_1, \dots, q_k\}\).
Note that the horizontal and vertical strip induced by each segment is now empty of points, while the axis-aligned quadrants starting at each corner of its bounding box may contain points; see \cref{fig:l1-exact-regions} for an illustration.
Let \(Q_1\) be the set of segments in \(Q\) with slope between \(-1\) and \(1\), and let \(Q_2\) be \(Q \setminus Q_1\).

We can now label the quadrants of points for each segment \(q_j\): let \(X_{j,1}\) and \(X_{j,2}\) be the quadrants to the left of \(q_j\), with \(X_{j,1}\) being the quadrant starting at the leftmost endpoint of \(q_j\), and \(X_{j,2}\) being the other.
Similarly, let \(X_{j,3}\) be the quadrant starting at the rightmost endpoint of \(q_j\), and let \(X_{j,4}\) be the other quadrant on the right.
In case of a horizontal or vertical segment, \(X_{j,2}\) and \(X_{j,4}\) are simply merged into \(X_{j,1}\) and \(X_{j,3}\), and it does not matter if \(X_{j,1}\) is the top or bottom quadrant.

\begin{figure}
    \centering
    \includegraphics{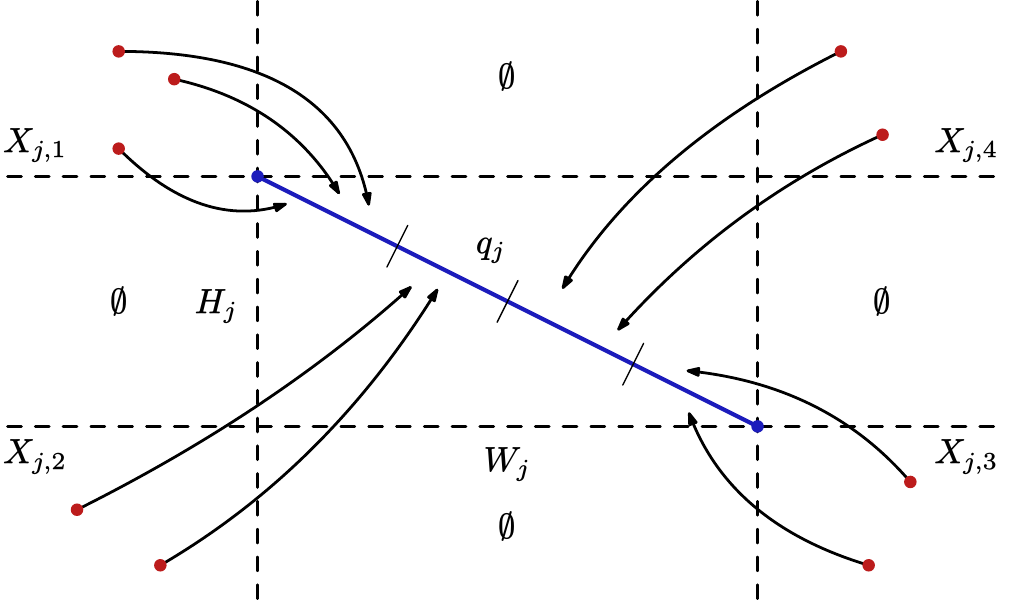}
    \caption{The regions of points for one subsegment \(q_j\), and the parts of the segment they assign their mass to.}
    \label{fig:l1-exact-regions}
\end{figure}

For all points in \(P\), the \(L_1\) distance to any point on the segment \(q_j\) is the same as the distance via one of the corners of the axis-aligned bounding box of \(q_j\).
Therefore, for each quadrant, we can separately consider the cost to reach the bounding box of \(q_j\) with a certain amount of mass, and the cost to spread that mass out over the segment.
Furthermore, the order in which a segment receives mass from the different quadrants in an optimal solution is fixed depending on its slope, see \cref{fig:l1-exact-regions}.
A simple swapping argument shows that the cost of an assignment not following this order can be decreased by making it follow the order.

Let \(u_{i,j}\) be the variable representing the amount of mass moved from \(p_i\) to \(q_j\), let \(d_{i,j}\) be the precomputed distance from \(p_i\) to the bounding box of \(q_j\), let \(W_j\) and \(H_j\) be the width and height of the bounding box of \(q_j\), and let \(w_j\) and \(h_j\) be constants such that \(w_j \cdot \ell\) is the absolute difference in \(x\)-coordinate when moving a distance of \(\ell\) along segment \(q_j\), and \(h_j \cdot \ell\) is the absolute difference in \(y\)-coordinate.
Writing \(\sum_{p_i \in X_{j,1}}{u_{i,j}}\) as \(x_{j,1}\) for convenience, for a given segment \(q_j \in Q_1\) we can write the cost \(c_{j,k}\) for moving the mass from the corner of the region \(X_{j,k}\) to the segment as follows:

\begin{linenomath*}
\begin{align*}
    c_{j,1} &= \tfrac{1}{2}x_{j,1}^2(w_j + h_j) \\
    c_{j,2} &= x_{j,2}(w_jx_{j,1} + w_j\tfrac{1}{2}x_{j,2} + (H_j - h_jx_{j,1} - h_j\tfrac{1}{2}x_{j,2})) \\
            &= x_{j,2}((w_j - h_j)(x_{j,1} + \tfrac{1}{2}x_{j,2}) + H_j) \\
    c_{j,3} &= \tfrac{1}{2}x_{j,3}^2(w_j + h_j) \\
    c_{j,4} &= x_{j,4}(w_jx_{j,3} + w_j\tfrac{1}{2}x_{j,4} + (H_j - h_jx_{j,3} - h_j\tfrac{1}{2}x_{j,4})) \\
            &= x_{j,4}((w_j - h_j)(x_{j,3} + \tfrac{1}{2}x_{j,4}) + H_j) \\
\end{align*}
\end{linenomath*}

\begin{figure}
    \centering
    \includegraphics{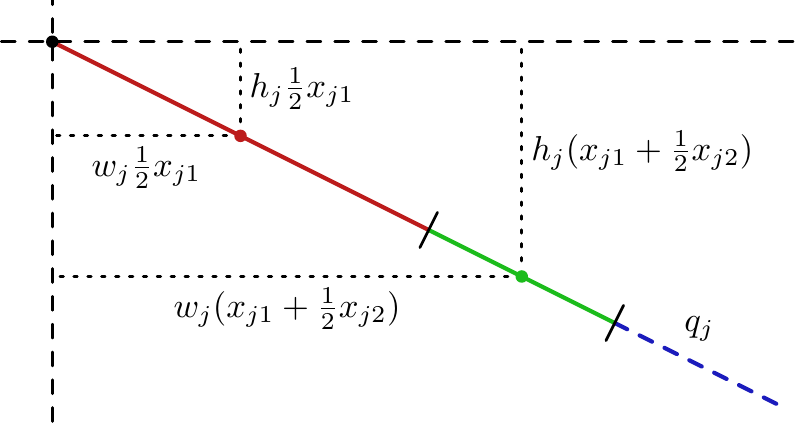}
    \caption{The calculations of the distances to the midpoints of the regions to which the mass from \(X_{j,1}\) and \(X_{j,2}\) will be assigned for a segment with slope between \(-1\) and \(1\).}
    \label{fig:l1-costs}
\end{figure}

Here we use the fact that under the \(L_1\) distance, the cost of sending mass to some connected region of a segment is the same as the cost of sending everything to the midpoint of the connected region; see \cref{fig:l1-costs} for an illustration of the calculations of the distances to these midpoints.
Note that we omit the cost of sending the mass from the points to the corners of the bounding box; this will be accounted for later.
Further note that \(w_j - h_j\) is always positive here.
Symmetrically, the costs \(c_k'(q_j)\) for a given segment \(q_j \in Q_2\) are as follows:

\begin{linenomath*}
\begin{align*}
    c'_{j,1} &= \tfrac{1}{2}x_{j,1}^2(w_j + h_j) \\
    c'_{j,2} &= x_{j,2}((h_j - w_j)(x_{j,3} + \tfrac{1}{2}x_{j,2}) + W_j) \\
    c'_{j,3} &= \tfrac{1}{2}x_{j,3}^2(w_j + h_j) \\
    c'_{j,4} &= x_{j,4}((h_j - w_j)(x_{j,1} + \tfrac{1}{2}x_{j,4}) + W_j)
\end{align*}
\end{linenomath*}

Note that \(h_j - w_j\) is always positive here.
We can now formalise our problem as follows:

\begin{linenomath*}
\begin{equation*}
    \tp^* = \argmin_{\mathbf{u}} \left(\sum_{j}{\sum_{i}{d_{i,j} \cdot u_{i,j}}}\right) + \sum_{k}{\left(\sum_{q_j \in Q_1}{c_{j,k}} + \sum_{q_j \in Q_2}{c'_{j,k}}\right)}
\end{equation*}
\end{linenomath*}
\begin{linenomath*}
\begin{equation*}
    \begin{alignedat}{3}
        & \text{subject to} \quad && u_{i,j} \geq 0                          && \qquad\qquad\quad\forall i, j \\
        &                         && \sum_{j}{u_{i,j}} = \norm{p_i}          && \qquad\qquad\quad\forall i    \\
        &                         && \sum_{i}{u_{i,j}} = \norm{q_j}          && \qquad\qquad\quad\forall j    \\
    \end{alignedat}
\end{equation*}
\end{linenomath*}

Since all \(d_{i,j}\), \(c_{j,k}\) and \(c'_{j,k}\) are non-negative, this is a sum of convex quadratic functions, giving a quadratic program with a convex objective function.

\begin{theorem}\label{thm:exact-algorithm}
    Let \(P\) be a set of \(n\) weighted points and \(S\) be a set of \(m\) line segments with equal total weight.
    It is possible to construct an exact optimal transport plan between \(P\) and \(S\) under the \(L_1\) metric by solving a convex quadratic program.
\end{theorem}

When all the weights in our objective function and constraints are integers, a convex quadratic program can be solved in weakly polynomial time, see e.g.~\cite{goldfarb1990qp,kozlov1980,monteiro1989qp}.
In our case, some of the weights may be real numbers.
In particular, \(w_j\) and \(h_j\) may be square roots of rational numbers.
It is not clear if such a program can be solved in polynomial time.

As square roots appear in many geometric settings, we are typically happy to assume a model of computation in which we can perform elementary operations on arbitrary real numbers in constant time.
However, even in such a model of computation, the typical methods (cited above) for solving convex quadratic programs in polynomial time may fail.
These methods generally rely on approximately solving a series of quadratic programs with increasing precision, and then argue that when the precision is high enough, the approximate solution can be rounded to the globally optimal solution.
The argument that such rounding works eventually relies on the input being integral.

This problem can be addressed in several ways.
First, we can employ different methods for solving the quadratic program, such as the simplex algorithm.
This method takes exponential time in the worst case, but has been shown to be polynomial in practice through smoothed analysis~\cite{dadush2020smoothed}.
Second, we can forego an exact algorithm and obtain a \((1 + \eps)\)-approximation by simply rounding the square roots in our program with enough precision.
Given the value of \(\eps\), it suffices to simply approximate the values such that the ratio of the rounded value to the original is at most \((1 + \eps)\).
Third, we could apply the \(L_1\) metric not only to distances, but also to the length of each segment.
If we define the length of a segment to be equal to the \(L_1\) distance between its endpoints, the equations for \(c_{j,k}\) and \(c'_{j,k}\) simplify significantly, and the square roots disappear.
Our entire program can then be made to have integer coefficients by simply requiring all points in \(P\) and endpoints of segments of \(S\) to lie on the integer lattice.
Note that this solution also applies if we restrict ourselves to classes of segments for which no square roots show up in our program, such as segments that are axis-aligned.

It may seem that our exact algorithm only runs in provably polynomial time in quite restricted cases.
However, it is worth noting here that it is not clear that an exact algorithm for the \(L_2\) case even exists, let alone one that runs in polynomial time.

\section{Points to segments}\label{sec:point-segment}
We now describe a polynomial-time algorithm that finds a transport plan with a cost that is at most \(1 + \eps\) times the cost of the optimal transport plan when one set consists of points with total weight one and the other of line segments with total length one.
The main idea is to reduce our instance to a transport problem on two weighted sets of points.
Our strategy is as follows: we subdivide each segment such that for each subsegment \(s'\) the ratio of the distance to the closest and furthest point on \(s'\) for every \(p_i \in P\) is at most \(1 + \delta\) for some appropriate choice of \(\delta \in O(\eps)\).
We then approximate a minimum cost flow problem on a bipartite graph between \(P\) and the subdivided segments, where the cost of any edge is equal to the shortest distance between a point and a subsegment.
Finally, we use the solution to this flow problem to build a discrete transport plan.
For an appropriate choice of \(\delta\), this gives a \((1 + \eps)\)-approximation.

The naive approach to subdividing the segments would be to make all the pieces some equal, appropriately small length.
However, we can reduce the number of subsegments required by subdividing the segments as follows\footnote{This reduces the total number of subsegments required from \(O(nm/\eps^2)\) to \(O(\frac{nm}{\eps}\log{\frac{1}{\eps}})\).}.
We repeatedly perform the following procedure for each subsegment.
If there exists a point in \(P\) such that the entire subsegment lies within distance \(\delta / nm\) of that point, do nothing.
Otherwise, if there is a point in \(P\) for which the ratio of the longest and shortest distance between that point and the current subsegment is more than \(1 + \delta\), cut the subsegment in half.
Call the resulting set of subsegments \(Q\); see \cref{fig:segment-subdivision-example} for an example.

\begin{figure}
    \centering
    \includegraphics{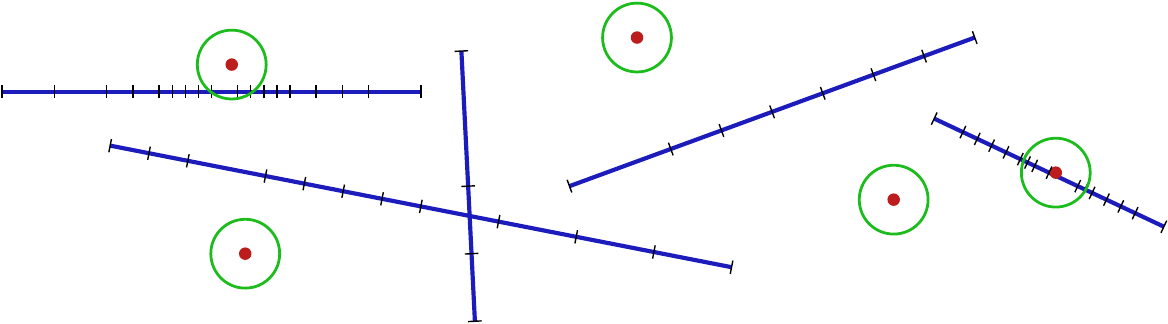}
    \caption{An example subdivision of a set of segments.
    Small perpendicular line segments delimit the generated subsegments.
    A green circle denotes the distance of \(\delta / nm\) from each point.
    Note that \(\eps\) is set to a very large value here for the clarity of the resulting image.}
    \label{fig:segment-subdivision-example}
\end{figure}

We now define a complete bipartite graph \(G = (P \cup Q, P \times Q)\), with edges between each point-subsegment pair (note that this graph is used for analysis only; our algorithm does not construct it).
The cost of each edge will simply be the shortest distance between the point and segment it connects.
A solution to a flow problem in \(G\) can be transformed into a transport plan by assigning a piece of subsegment to a point with length equal to the amount of flow along the corresponding edge.
We will show that the EMD between \(P\) and \(S\) is approximated by the cost of any transport plan derived from a minimum cost flow in \(G\).

First note the following general lower bound on the cost of an optimal solution:

\begin{lemma}\label{lem:lower-bound-flow}
    The earth mover's distance \(\norm{\tp^*}\) between \(P\) and \(Q\) is bounded from below by the cost \(\norm{\mathcal{W}}\) of a minimum cost flow \(\mathcal{W}\) in \(G\).
\end{lemma}
\begin{proof}
    Consider any transport plan \(\tp^*\) that minimises the earth mover's distance.
    If, for each point \(p_i\) that moves mass to some segment \(q_j\), we modify the transport plan such that the mass is moved only to the point on \(q_j\) closest to \(p_i\), we obtain a plan \(\lambda\) with cost \(\norm{\lambda} \leq \norm{\tp^*}\).
    Such a plan is a solution to a flow problem in \(G\), as it moves all available mass.
    It follows that the cost \(\norm{\mathcal{W}}\) of a minimum cost flow \(\mathcal{W}\) in \(G\) satisfies \(\norm{\mathcal{W}} \leq \norm{\tp^*}\).
    \hfill \qed
\end{proof}

We also note the following lower bound on the value of \(\norm{\mathcal{W}}\):

\begin{lemma}\label{lem:lower-bound-subdivision-point-segment}
    \(\norm{\mathcal{W}} \geq \dfrac{\delta - 2\delta^2 - 2\delta^3}{nm}\).
\end{lemma}
\begin{proof}
    For a given point-segment pair \((p, s) \in P \times S\), consider the segments in \(Q\) derived from \(s\) that have a point within distance \(\delta / nm\) of \(p\).
    By construction, such a segment has its furthest point at distance at most \((1 + \delta) \cdot \delta / nm = \delta / nm + \delta^2 / nm\) to \(p\).
    Therefore, the total length of these segments is at most \(2(\delta / nm + \delta^2 / nm)\) for a given \(p\) and \(s\).
    Over all point-segment pairs, this gives a total length of at most \(2\delta + 2\delta^2\).
    This means the total length of segments in \(Q\) with distance to the closest point in \(P\) at least \(\delta / nm\) is at least \(1 - 2\delta - 2\delta^2\).
    The cost of a minimum flow in \(\mathcal{W}\) is therefore at least \((1 - 2\delta - 2\delta^2) \cdot \delta / nm = (\delta - 2\delta^2 - 2\delta^3) / nm\).
    \hfill \qed
\end{proof}

We calculate a transport plan \(\tp\) between \(P\) and \(Q\) as follows.
First, we approximate each segment \(q \in Q\) by a point somewhere on that segment with weight equal to the length of \(q\); call this set of points \(T\).
We obtain \(\tp\) by calculating an optimal transport plan \(\nu\) between \(P\) and \(T\), and then spreading the mass sent to each point \(t \in T\) evenly over the segment in \(Q\) that point was derived from.
We now bound the cost of \(\tp\) in terms of \(\norm{\mathcal{W}}\):

\begin{lemma}\label{lem:bound-discretised-point-segment}
    \(\norm{\mathcal{W}} \leq \norm{\tp} \leq (1 + \delta)^2 \norm{\mathcal{W}} + \dfrac{4\delta^2}{nm} + \dfrac{2\delta^3}{nm}\).
\end{lemma}
\begin{proof}
    We first bound the cost of \(\nu\).
    In \(\mathcal{W}\), we measured all the distances to the closest point on each subsegment.
    Imagine that we picked all the points in \(T\) to be the furthest point on the subsegment.
    For the subsegments with furthest distance to a point in \(P\) of at least \(\delta / nm\), the ratio of these distances is at most \(1 + \delta\) by construction.
    We can therefore bound all the parts of \(\nu\) where the furthest distance to a point in \(P\) is at least \(\delta / nm\) by \((1 + \delta)\norm{\mathcal{W}}\).
    The total mass being moved over distance at most \(\delta / nm\) in \(\nu\) is at most \(2\delta\), giving a cost of at most \(2\delta^2 / nm\).
    The total cost of \(\nu\) is therefore at most \((1 + \delta)\norm{\mathcal{W}} + 2\delta^2 / nm\).

    Now consider the extra cost incurred when transforming \(\nu\) into \(\tp\) by spreading the mass out evenly over all the segments.
    We can use the same argument as before: for the parts of \(\nu\) with distance to a point in \(P\) of at least \(\delta / nm\), the cost increases by a factor of at most \(1 + \delta\), and the total cost of the part within distance \(\delta / nm\) is at most \(2\delta^2 / nm\).
    We can therefore bound the cost of \(\tp\) by \((1 + \delta)\norm{\nu} + 2\delta^2 / nm\).

    We now obtain the upper bound stated in the lemma by plugging the bound on \(\nu\) into the bound on \(\tp\).
    The lower bound follows directly from the fact that none of the distances in \(\tp\) are smaller than the distances between the same objects in \(\mathcal{W}\).
    \hfill \qed
\end{proof}

We now show that \(\norm{\tp}\) approximates \(\norm{\tp^*}\).

\begin{theorem}\label{thm:point-segment-approximation}
    \(\norm{\tp}\) is a \((1 + 17\delta)\)-approximation to the earth mover's distance \(\norm{\tp^*}\) between \(P\) and \(S\) for \(0 < \delta \leq \frac{1}{4}\).
\end{theorem}
\begin{proof}
    By \cref{lem:bound-discretised-point-segment}, we know that
    \begin{linenomath*}
    \begin{equation*}
        \norm{\tp} \leq (1 + \delta)^2\norm{\mathcal{W}} + \frac{4\delta^2}{nm} + \frac{2\delta^3}{nm}
    \end{equation*}
    \end{linenomath*}

    \(\norm{\mathcal{W}}\) is also a lower bound on \(\norm{\tp}\); the ratio between the upper and lower bound is
    \begin{linenomath*}
    \begin{equation*}
        \frac{(1 + \delta)^2\norm{\mathcal{W}} + \frac{4\delta^2}{nm} + \frac{2\delta^3}{nm}}{\norm{\mathcal{W}}}
    \end{equation*}
    \end{linenomath*}

    This ratio is the largest for small values of \(\norm{\mathcal{W}}\), so we plug in the lower bound from \cref{lem:lower-bound-subdivision-point-segment}:
    \begin{linenomath*}
    \begin{align*}
                & \quad \frac{(1 + \delta)^2 \cdot \frac{\delta - 2\delta^2 - 2\delta^3}{nm} + \frac{4\delta^2}{nm} + \frac{2\delta^3}{nm}}{\frac{\delta - 2\delta^2 - 2\delta^3}{nm}} \\
        =       & \quad \frac{1 + 4\delta - 3\delta^2 - 6\delta^3 - 2\delta^4}{1 - 2\delta - 2\delta^2} \\
        =       & \quad 1 + \delta + \frac{5\delta + \delta^2 - 4\delta^3 - 2\delta^4}{1 - 2\delta - 2\delta^2} \\
        \leq    & \quad 1 + \delta + \frac{6\delta}{1 - 2\delta - 2\delta^2} \\
        \leq    & \quad 1 + 17\delta \tag{assuming \(\delta \leq \frac{1}{4}\)}
    \end{align*}
    \end{linenomath*}

    As \(\norm{\mathcal{W}}\) is also a lower bound for \(\norm{\tp^*}\) (\cref{lem:lower-bound-flow}), and \(\tp\) can obviously not have lower cost than the optimal transport plan, this gives a \((1 + 17\delta)\)-approximation.
    \hfill \qed
\end{proof}

Setting \(\delta = \eps / 17\) gives a \((1 + \eps)\)-approximation.
Note that the bound on \(\delta\) is not restrictive: for any constant \(\eps\) that would require a larger value of \(\delta\), we can simply use the value \(1 / 4\) at the cost of a constant factor in the running time of our algorithm.

\begin{figure}
    \centering
    \includegraphics{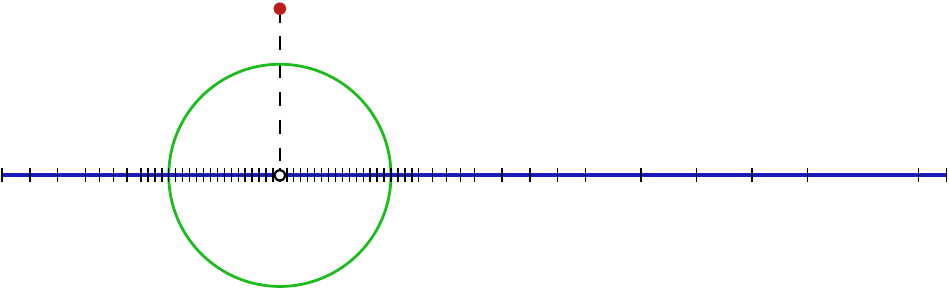}
    \caption{An example of a single \(R_{i,j}\).
    Small perpendicular segments delimit the generated subsegments.
    The green circle denotes the distance of \(\delta / nm\) from the projected point.}
    \label{fig:segment-alternative-count-example}
\end{figure}

\subsection{Running time analysis}\label{sec:point-segment-analysis}
We now analyse the number of subsegments in \(Q\).
We will count the number of subsegments in a different subdivision of \(S\), and then show that \(Q\) has at most a constant factor more subsegments.
The alternative subdivision of each \(s_j\) will be as follows: project each \(p_i\) onto the supporting line of \(s_j\), call this point \(p_{i,j}\).
We construct the one-dimensional Voronoi diagram of all \(p_{i,j}\) along the supporting line of \(s_j\); let \(s_{i,j}\) be the part of \(s_j\) inside the Voronoi cell of \(p_{i,j}\).
From each \(p_{i,j}\), we subdivide \(s_{i,j}\) into both directions.
Up to a distance of \(\delta / nm\), we make subsegments of size \(\delta^2 / nm\).
Moving outward, we double the size of the subsegments whenever their ratio of distances to \(p_{i,j}\) would still be below \(1 + \delta\).
Let \(R_{i,j}\) be the resulting subdivision; see \cref{fig:segment-alternative-count-example} for an example.

\begin{lemma}\label{lem:point-segment-alternative-count}
    \(R = \bigcup R_{i,j}\) has \(O\left(\dfrac{nm}{\delta}\log\dfrac{1}{\delta}\right)\) subsegments.
\end{lemma}
\begin{proof}
    We define \(\beta = \frac{\delta}{nm}\) and \(\gamma = \frac{\delta^2}{nm}\).
    In the following, we analyse only the case where \(p_{i,j}\) is on \(s_{i,j}\); if it lies outside, the number of subsegments will be smaller, as the size of the subsegments increases with distance.
    The length covered as we add subsegments on \(s_{i,j}\) can be written as
    \begin{linenomath*}
    \begin{equation*}
        \beta + 2\sum_{i = 0}^{k}{\alpha_i 2^i \gamma}
    \end{equation*}
    \end{linenomath*}

    \noindent
    where \(k\) is the number of times we double the size of the subsegments, and \(\alpha_i\) is the number of subsegments with a size that has been doubled \(i\) times.
    The number of subsegments can then be calculated by finding the values of \(k\) and \(\alpha_i\).
    We start with \(\alpha_0\), which can be found by considering the distance at which the next cell could be double the size:
    \begin{linenomath*}
    \begin{align*}
        \frac{\beta + \alpha_0\gamma + 2\gamma}{\beta + \alpha_0\gamma} & \leq 1 + \delta                                   \\
        \frac{2\gamma}{\beta + \alpha_0\gamma}                          & \leq \delta                                       \\
        \alpha_0                                                        & \geq \frac{2\gamma - \delta\beta}{\delta\gamma} = \frac{2}{\delta} - \frac{\beta}{\gamma} = \frac{1}{\delta} \\
        \alpha_0                                                        & \geq \frac{1}{\delta} \\
    \end{align*}
    \end{linenomath*}

    Per the procedure described above, we double the size of the subsegments as soon as this is allowed.
    This corresponds to taking the values of \(\alpha_i\) as small as possible, so we take \(\alpha_0 = \frac{1}{\delta}\).
    Next, we can show by induction that all \(\alpha_i\) are equal:
    \begin{linenomath*}
    \begin{equation*}
        \textnormal{IH: } \alpha_j = \alpha^* = \frac{1}{\delta} \textnormal{ for } j < i.
    \end{equation*}
    \end{linenomath*}%
    \begin{linenomath*}
    \begin{align*}
        \frac{\beta + 2^{i + 1}\gamma + \sum_{j = 0}^{i}{\alpha_j 2^j \gamma}}{\beta + \sum_{j = 0}^{i}{\alpha_j 2^j \gamma}}   & = 1 + \delta                                                                                           \\
        \frac{2^{i + 1}\gamma}{\beta + \sum_{j = 0}^{i}{\alpha_j 2^j \gamma}}                                                   & = \delta                                                                                               \\
        2^{i + 1}\gamma                                                                                                         & = \delta\beta + \delta\sum_{j = 0}^{i}{\alpha_j 2^j\gamma}                                             \\
& = \delta\beta + \delta\alpha_i 2^i\gamma + \delta\sum_{j = 0}^{i - 1}{\alpha_j 2^j\gamma}              \\
        2^{i + 1}                                                                                                               & = 1 + \delta\alpha_i 2^i + \delta\alpha^*(2^i - 1)                                                     \\
        \alpha_i                                                                                                                & = \frac{2^{i + 1}}{\delta 2^i} - \frac{1}{\delta 2^i} - \frac{\delta\alpha^* (2^i - 1)}{\delta 2^i}    \\
& = \frac{2}{\delta} - \frac{1}{\delta 2^i} - \frac{2^i - 1}{\delta 2^i}                                 \\
& = \frac{2}{\delta} - \frac{1}{\delta}                                                                  \\
& = \frac{1}{\delta}                                                                                     \\
    \end{align*}
    \end{linenomath*}

    Knowing that all \(\alpha_i\) are equal to \(\frac{1}{\delta}\), we can determine the value of \(k\):
    \begin{linenomath*}
    \begin{align*}
        \beta + \sum_{i = 0}^{k}{\frac{1}{\delta}2^i\gamma} & \geq \norm{s_{i,j}}                        \\
        \beta + \frac{1}{\delta}\gamma(2^{k + 1} - 1)       & \geq \norm{s_{i,j}}                        \\
        \beta + \beta(2^{k + 1} - 1)                        & \geq \norm{s_{i,j}}                        \\
        2^{k + 1}                                           & \geq \frac{\norm{s_{i,j}} }{\beta}         \\
        k                                                   & \geq \log\frac{\norm{s_{i,j}} }{\beta} - 1 \\
    \end{align*}
    \end{linenomath*}

    This gives a total number of subsegments of \(O(\frac{1}{\delta}\log\frac{\norm{s_{i,j}}}{\beta})\) for each point-segment pair.
    The sum over all pairs is largest when all \(\norm{s_{i,j}}\) are equal, i.e.
    \(1 / nm\).
    This gives us a total number of subsegments for all pairs of \(O\left(\frac{nm}{\delta}\log\frac{1 / nm}{\beta}\right) = O\left(\frac{nm}{\delta}\log\frac{1}{\delta}\right)\).
    \hfill \qed
\end{proof}

\begin{lemma}\label{lem:point-segment-count}
    The set \(Q\) has \(O\left(\dfrac{nm}{\delta}\log\dfrac{1}{\delta}\right)\) subsegments.
\end{lemma}
\begin{proof}
    Consider any subsegment \(r \in R\).
    Any subsegment \(q \in Q\) that overlaps with \(r\) has \(\norm{q} \geq \norm{r}/4\): otherwise \(q\) was subdivided unnecessarily.
    As the subsegments in \(Q\) are disjoint, it follows that \(r\) can overlap with at most 5 subsegments in \(Q\).
    As such, \(Q\) contains at most 5 times more subsegments than \(R\), which, by \cref{lem:point-segment-alternative-count}, is \(O\left(\frac{nm}{\delta}\log\frac{1}{\delta}\right)\).
    \hfill \qed
\end{proof}

Putting everything together, we obtain the following result:

\begin{theorem}\label{thm:point-segment-main-result}
    Let \(P\) be a set of \(n\) weighted points and \(S\) be a set of \(m\) line segments with equal total weight, let \(\norm{\tp^*}\) be the cost of an optimal transport plan between them, and let \(\delta\) be any constant \(> 0\).
    Given an algorithm that constructs a \((1 + \delta)\)-approximation between weighted sets of \(k\) points in \(f_\delta(k)\) time, we can construct a transport plan between \(P\) and \(S\) with cost \(\leq (1 + 25\delta)\norm{\tp^*}\) in \(O\left(f_\delta\left(\frac{nm}{\delta}\log\left(\frac{1}{\delta}\right)\right)\right)\) time.
\end{theorem}
\begin{proof}
    In \cref{thm:point-segment-approximation}, we prove that an optimal transport plan \(\nu\) between \(P\) and \(T\) approximates \(\norm{\tp^*}\).
    However, we may be able to compute a \((1 + \delta)\)-approximation to \(\nu\) faster than we are able to compute it exactly.
    It remains to be shown that this approximation also suffices.

    Plugging in a \((1 + \delta)\)-approximation to \(\norm{\nu}\), rather than the exact value, we obtain the ratio
    \begin{linenomath*}
    \begin{equation*}
        \frac{(1 + \delta)^3\norm{\mathcal{W}} + \frac{4\delta^2}{nm} + \frac{4\delta^3}{nm} + \frac{2\delta^4}{nm}}{\norm{\mathcal{W}}}
    \end{equation*}
    \end{linenomath*}

    Following the same strategy as in the proof of \cref{thm:point-segment-approximation}, we derive the approximation ratio as follows:
    \begin{linenomath*}
    \begin{align*}
                & \quad \frac{(1 + \delta)^3 \cdot (\delta - 2\delta^2 - 2\delta^3) + 4\delta^2 + 4\delta^3 + 2\delta^4}{\delta - 2\delta^2 - 2\delta^3}    \\
        =       & \quad \frac{(1 + 3\delta + 3\delta^2 + \delta^3)(1 - 2\delta - 2\delta^2) + 4\delta + 4\delta^2 + 2\delta^3}{1 - 2\delta - 2\delta^2}     \\
        =       & \quad \frac{1 + 5\delta - \delta^2 - 9\delta^3 - 8\delta^4 - 2\delta^5}{1 - 2\delta - 2\delta^2}                                          \\
        =       & \quad 1 + \frac{7\delta + \delta^2 - 9\delta^3 - 8\delta^4 - 2\delta^5}{1 - 2\delta - 2\delta^2}                                          \\
        =       & \quad 1 + \delta + \frac{6\delta + 3\delta^2 - 7\delta^3 - 8\delta^4 - 2\delta^5}{1 - 2\delta - 2\delta^2}                                \\
        \leq    & \quad 1 + \delta + \frac{9\delta}{1 - 2\delta - 2\delta^2}                                                                                \\
        \leq    & \quad 1 + 25\delta                                                                                                                        \tag{assuming \(\delta \leq \frac{1}{4}\)}
    \end{align*}
    \end{linenomath*}

    As such, using an approximation of \(\nu\) still gives us an approximation of \(\tp^*\), albeit with a somewhat worse dependency on \(\delta\).
    \hfill \qed
\end{proof}

To our knowledge, the current fastest algorithm to calculate a \((1 + \delta)\)-approximation to \(\nu\) is that by Fox and Lu~\cite{fox2022}, which runs in \(O(N\delta^{-O(1)}\polylog{N})\) time, where \(N\) is the size of the input.
Setting \(\delta = \eps / 25\), this gives the following corollary to the previous theorem:

\begin{corollary}\label{cor:point-segment-main-result}
    For any constant \(\eps > 0\), a transport plan between \(P\) and \(S\) with cost \(\leq (1 + \eps)\norm{\tp^*}\) can be constructed in \(O\left(\frac{nm}{\eps^c} \polylog\left(\frac{nm}{\eps}\right)\right)\) time with high probability.
\end{corollary}

\section{Points to triangles}\label{sec:point-triangle}
\begin{figure}
    \centering
    \includegraphics{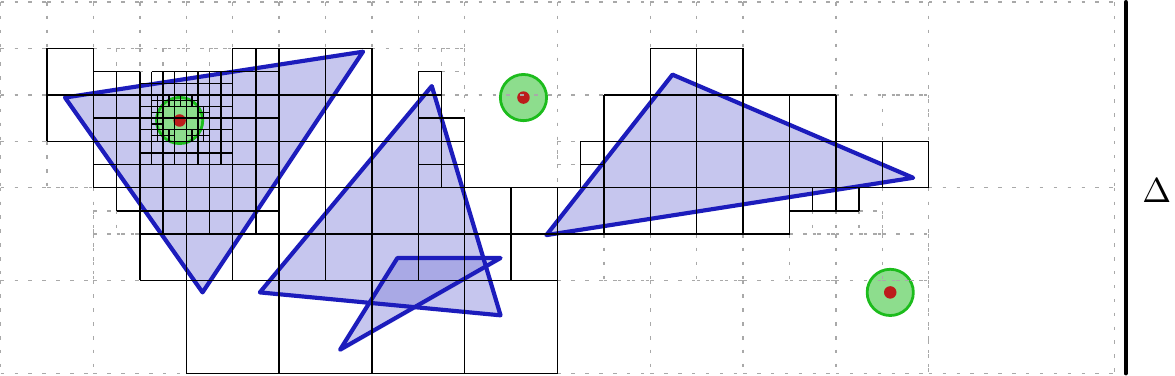}
    \caption{An example subdivision of a set of triangles.
    Each cell records the total area of triangles it intersects.
    Cells that are part of \(Q\) are shown in black; empty cells are shown in grey dashed lines.
    A green disk denotes the distance of \(\delta / \sqrt{nm}\) from each point.
    Note that \(\eps\) is set to a very large value here for the clarity of the resulting image.}
    \label{fig:triangle-subdivision-example}
\end{figure}

We consider the case where \(P\) is a set of weighted points with total weight one and \(S\) is a set of \(m\) triangles with total area one.
We denote the longest edge of any triangle by \(\Delta\).
Our strategy is similar to before: we subdivide the triangles such that for each subregion, the ratio between its shortest and longest distance to each point is at most \(1 + \delta\) for some appropriate choice of \(\delta \in O(\eps)\).
We then show that a solution based on an optimal transport plan between \(P\) and some points inside the subregions approximates an optimal solution.

We first overlay a uniform grid onto our triangles with grid cells of size \(\Delta \times \Delta\).
We can identify the cells of this grid that contain a triangle in \(O(m \log m)\) time using point-location in a compressed quadtree where the smallest cell size is \(\Delta \times \Delta\)~\cite{vanderhoog2018}.
As each triangle can intersect at most four cells, the total size of this set of cells is \(O(m)\).
We now recursively subdivide each cell as follows: if there is a point in \(P\) such that the whole cell is within distance \(\delta / \sqrt{nm}\) of it, we stop; otherwise, if for any point the ratio of distances to the furthest and closest point in this cell is more than \(1 + \delta\), we subdivide this cell into four cells of one quarter the area.
If the ratio holds for all points, we stop.
Call the resulting set of cells \(Q\); see \cref{fig:triangle-subdivision-example} for an example.

During each subdivision, we keep track of the total area of triangles contained inside that cell.
We can then once again build a complete bipartite graph \(G = (P \cup Q, P \times Q)\), with the capacity of each vertex set to the weight of the corresponding point or the total area of triangles contained in the corresponding cell, and the weight of each edge equal to the shortest distance between the point and the cell it connects.
The cost of a minimum cost flow \(\mathcal{W}\) is now once again a lower bound to the EMD, exactly as in \cref{lem:lower-bound-flow}.
In an analogous way to \cref{lem:lower-bound-subdivision-point-segment}, we obtain a lower bound on the cost of \(\mathcal{W}\):

\begin{lemma}\label{lem:lower-bound-subdivision-point-triangle}
    \(\norm{\mathcal{W}} \geq \dfrac{\delta}{\sqrt{nm}} - \dfrac{\pi\delta^3 + 2\pi\delta^4 + \pi\delta^5}{\sqrt{nm}}\).
\end{lemma}
\begin{proof}
    For a given point-triangle pair \((p, s) \in P \times S\), consider the cells in \(Q\) intersecting \(s\) that have a point within distance \(\delta / \sqrt{nm}\) of \(p\).
    By construction, such a cell has its furthest point at distance at most \((1 + \delta) \cdot \delta / \sqrt{nm} = \delta / \sqrt{nm} + \delta^2 / \sqrt{nm}\).
    Therefore, the total area of these cells is at most \(\pi(\delta / \sqrt{nm} + \delta^2 / \sqrt{nm})^2 = \pi(\delta^2 + 2\delta^3 + \delta^4) / nm\).
    Over all point-triangle pairs, this gives a total area of at most \(\pi(\delta^2 + 2\delta^3 + \delta^4)\).
    This leaves \(1 - \pi(\delta^2 + 2\delta^3 + \delta^4)\) with distance at least \(\delta / \sqrt{nm}\) in \(\mathcal{W}\).
    The cost is therefore at least \((1 - \pi(\delta^2 + 2\delta^3 + \delta^2)) \cdot \delta / \sqrt{nm} = \delta / \sqrt{nm} - \pi(\delta^3 + 2\delta^4 + \delta^5) / \sqrt{nm}\).
    \hfill \qed
\end{proof}

We now once again approximate \(\norm{\mathcal{W}}\) by reducing the flow problem to a transportation problem between two sets of weighted points.
Again, we pick any point in each cell \(q \in Q\) and give it a weight equal to the area of triangles contained in \(q\); call this set of points \(T\).

\begin{lemma}\label{lem:point-triangle-bounds}
    \(\norm{\mathcal{W}} \leq \norm{\tp} \leq (1 + \delta)^2\norm{\mathcal{W}} + \frac{2\pi\delta^3}{\sqrt{nm}} + \frac{\pi\delta^4}{\sqrt{nm}}\)
\end{lemma}
\begin{proof}
    Let \(\nu\) be an optimal transport plan between \(P\) and \(T\), and let \(\norm{\nu}\) be its cost.
    We can upper bound \(\norm{\nu}\) by measuring all distances to the furthest point in each cell.
    We constructed \(Q\) such that the ratio of the closest and furthest distance between any point-cell pair is \(1 + \delta\) when the furthest distance is at least \(\delta / \sqrt{nm}\).
    We can therefore bound all parts of \(\nu\) where the distance is at least \(\delta / \sqrt{nm}\) by \((1 + \delta)\norm{\mathcal{W}}\).
    The total mass being moved over a distance at most \(\delta / \sqrt{nm}\) in \(\nu\) is at most \(\pi\delta^2\), giving a cost of \(\pi\delta^3 / \sqrt{nm}\).
    The total cost when measuring to the furthest point is therefore \((1 + \delta)\norm{\mathcal{W}} + \pi\delta^3 / \sqrt{nm}\).

    We now turn \(\nu\) into a transport plan \(\tp\) between \(P\) and \(Q\) by spreading the mass sent to each point \(t \in T\) out evenly over the parts of the triangles in the cell in \(Q\) that \(t\) was derived from.
    By construction, for cells with a distance of at least \(\delta / \sqrt{nm}\), this increases the cost by at most a factor \(1 + \delta\).
    We can therefore bound the cost of this part of \(\tp\) by \((1 + \delta)\norm{\nu}\).
    The remaining part has a total mass of at most \(\pi\delta^2\), giving a cost of \(\pi\delta^3 / \sqrt{nm}\).
    The total cost of \(\tp\) is then bound by \((1 + \delta)\norm{nu} + \pi\delta^3 / \sqrt{nm}\).

    Plugging in the bound on \(\norm{\nu}\) obtained above, we obtain an upper bound of \((1 + \delta)^2\norm{\mathcal{W}} + 2\pi\delta^3 / \sqrt{nm} + \pi\delta^4 / \sqrt{nm}\).
    The lower bound follows directly from the fact that none of the distance in \(\tp\) are smaller than the distances between the same objects in \(\mathcal{W}\).
    \hfill \qed
\end{proof}

Putting this all together, we can show that \(\norm{\tp}\) approximates \(\norm{\tp^*}\).

\begin{theorem}\label{thm:point-triangle-approximation}
    \(\norm{\tp}\) is a \((1 + 9\delta)\)-approximation to the earth mover's distance \(\norm{\tp^*}\) between \(P\) and \(S\) for \(0 < \delta \leq \frac{1}{2\pi}\).
\end{theorem}
\begin{proof}
    By \cref{lem:point-triangle-bounds} have that
    \begin{linenomath*}
    \begin{equation*}
        \norm{\tp} \leq (1 + \delta)^2\norm{\mathcal{W}} + \frac{2\pi\delta^3}{\sqrt{nm}} + \frac{\pi\delta^4}{\sqrt{nm}}
    \end{equation*}
    \end{linenomath*}

    \(\norm{\mathcal{W}}\) is also a lower bound on \(\norm{\tp}\); the ratio between the upper and lower bound is
    \begin{linenomath*}
    \begin{equation*}
        \frac{(1 + \delta)^2\norm{\mathcal{W}} + \frac{2\pi\delta^3}{\sqrt{nm}} + \frac{\pi\delta^4}{\sqrt{nm}}}{\norm{\mathcal{W}}}
    \end{equation*}
    \end{linenomath*}

    This ratio is largest for small values of \(\norm{\mathcal{W}}\), so we plug in the lower bound from \cref{lem:lower-bound-subdivision-point-triangle}:
    \begin{linenomath*}
    \begin{align*}
                & \quad \frac{(1 + \delta)^2\norm{\mathcal{W}} + \frac{2\pi\delta^3}{\sqrt{nm}} + \frac{\pi\delta^4}{\sqrt{nm}}}{\norm{\mathcal{W}}}                                                                                                              \\
        \leq    & \quad \frac{(1 + \delta)^2\left(\frac{\delta}{\sqrt{nm}} - \frac{\pi\delta^3 + 2\pi\delta^4 + \pi\delta^5}{\sqrt{nm}}\right) + \frac{2\pi\delta^3}{\sqrt{nm}} + \frac{\pi\delta^4}{\sqrt{nm}}}{\frac{\delta}{\sqrt{nm}} - \frac{\pi\delta^3 + 2\pi\delta^4 + \pi\delta^5}{\sqrt{nm}}}  \\
        \leq    & \quad \frac{(1 + 2\delta + \delta^2)(1 - \pi\delta - 2\pi\delta^2 - \pi\delta^3) + 2\pi\delta^2 + \pi\delta^3}{1 - \pi\delta - 2\pi\delta^2 - \pi\delta^3}                                                                                \\
        =       & \quad \frac{1 + 2\delta + \delta^2 - 2\pi\delta^2 - 5\pi\delta^3 - 4\pi\delta^4 - \pi\delta^5}{1 - \pi\delta - 2\pi\delta^2 - \pi\delta^3}                                                                                                \\
        =       & \quad 1 + \frac{2\delta + \delta^2 - \pi\delta - 4\pi\delta^3 - 4\pi\delta^4 - \pi\delta^5}{1 - \pi\delta - 2\pi\delta^2 - \pi\delta^3}                                                                                                   \\
        =       & \quad 1 + \delta + \frac{\delta + \delta^2 - \pi\delta - 2\pi\delta^3 - 3\pi\delta^4 - \pi\delta^5}{1 - \pi\delta - 2\pi\delta^2 - \pi\delta^3}                                                                                           \\
        <       & \quad 1 + \delta + \frac{\delta^2}{1 - \pi\delta - 2\pi\delta^2 - \pi\delta^3}                                                                                                                                                            \\
        \leq    & \quad 1 + \delta + \frac{\delta^2}{1 - \frac{1}{2} - \frac{1}{\pi} - \frac{1}{2\pi^2}}                                                                                                                                                    \tag{assuming \(\delta \leq \frac{1}{2\pi}\)} \\
        <       & \quad 1 + \delta + 8\delta^2                                                                                                                                                                                                              \\
        <       & \quad 1 + 9\delta                                                                                                                                                                                                                         \\
    \end{align*}
    \end{linenomath*}

    As \(\norm{\mathcal{W}}\) is also a lower bound for \(\norm{\tp}\) (\cref{lem:lower-bound-flow}), and \(\tp\) can obviously not have lower cost than the optimal transport plan, this gives a \((1 + 9\delta)\)-approximation.
    \hfill \qed
\end{proof}

Setting \(\delta = \eps / 9\) gives a \((1 + \eps)\)-approximation.

\subsection{Running time analysis}\label{sec:point-triangle-analysis}
Our analysis will be the same as in \cref{sec:point-segment-analysis}; we just need to determine the size of \(Q\).
We will once again make an alternative subdivision of each \(s_j \in S\), count the number of cells in that subdivision, and then argue that \(\norm{Q}\) differs by at most a constant factor.
Our alternative subdivision is a direct adaptation of the one used in \cref{sec:point-segment-analysis} to two dimensions: for each point \(p_i\) and triangle \(s_j\), we fill a square with side length \(2\delta / \sqrt{nm}\) centred on \(p_i\) with cells of size \(\delta^2 / \sqrt{nm}\).
From there, we add rings of cells of side length \(\delta^2 / \sqrt{nm}\) around the square, until the next full ring could have cells double the size without violating the ratio of \(1 + \delta\) between the shortest and longest distance to \(p_i\) for any cell in the ring.
We repeat this process until we have covered a square of size \(\Delta \times \Delta\).
Let \(R_{i,j}\) be the resulting set of cells; see \cref{fig:point-quadtree-cells} for an example.
The proof is similar to \cref{lem:point-segment-alternative-count}.

\begin{figure}
    \centering
    \begin{subfigure}{0.6\textwidth}
        \raggedright
        \hspace*{10mm}
        \includegraphics{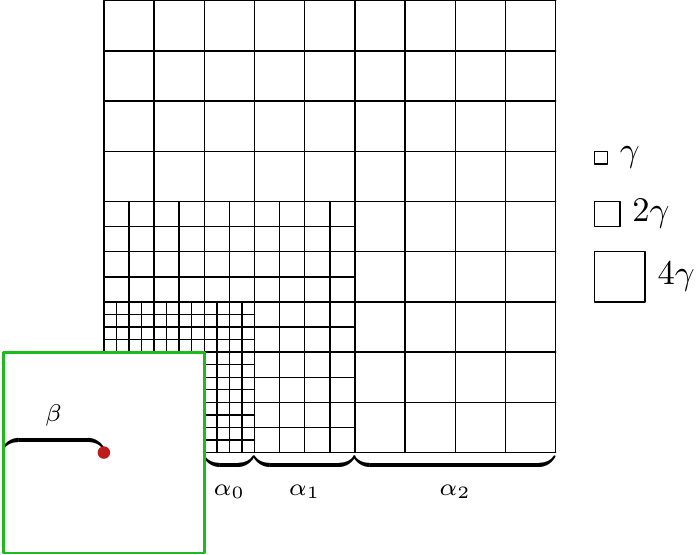}
    \end{subfigure}%
    \hfill
    \begin{subfigure}{0.35\textwidth}
        \raggedleft
        \includegraphics{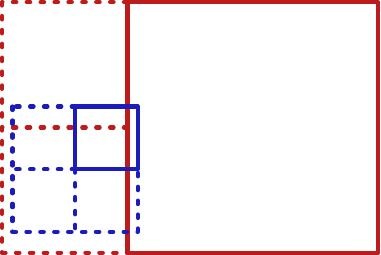}
        \hspace*{10mm}
    \end{subfigure}
    \caption{On the left, part of one quadrant of the construction of \(R_{i,j}\).
    There are \(\alpha_i\) layers of cells of size \(2^i\gamma\) before the size is doubled.
    On the right, an illustration of the argument that a cell of \(Q\) (blue) has at least one quarter the edge length of a cell of \(R\) that it intersects (red).}
    \label{fig:point-quadtree-cells}
\end{figure}

\begin{lemma}\label{lem:point-triangle-alternative-count}
    \(R = \bigcup R_{i,j}\) has \(O\left(\dfrac{nm}{\delta^2}\log\dfrac{nm\Delta}{\delta}\right)\) cells.
\end{lemma}
\begin{proof}
    We define \(\beta = \frac{\delta}{\sqrt{nm}}\) and \(\gamma = \frac{\delta^2}{\sqrt{nm}}\).
    In the following, we only analyse the case where \(p_{i,j}\) is inside \(s_{i,j}\); if it lies outside, the number of cells will be smaller, as the size of the cells increases with distance.
    We also analyse the number of cells in one quadrant only; the total number is simply four times as many.
    See \cref{fig:point-quadtree-cells} for an illustration of a quadrant.
    The number of cells created as we add rings of cells on \(s_{i,j}\) can then be written as
    \begin{linenomath*}
    \begin{equation*}
        \frac{\beta^2}{\gamma^2} + \sum_{i = 0}^{k}{2\alpha_i \cdot \frac{\beta + \sum_{j = 0}^{i - 1}{\alpha_j2^j\gamma}}{2^i\gamma} + \alpha_i^2}
    \end{equation*}
    \end{linenomath*}

    \noindent
    where \(k\) is the number of times we double the size of the cells, and \(\alpha_i\) is the number of rings containing cells of a size that has been doubled \(i\) times.
    The number of cells can then be calculated by finding the values of \(k\) and \(\alpha_i\).
    We take the values of \(\alpha_i\) to be the same as in \cref{lem:point-segment-alternative-count} (i.e.
    \(1 / \delta\)): along a horizontal or vertical line through \(p_i\) these values give the exactly correct distance ratios, and cells not on this line can be made to have the correct ratio through one extra subdivision.

    Let \((x, y)\) be the vector from \(p_i\) to the closest point on the cell.
    Assume w.l.o.g. that \(0 \leq y \leq x\); the other cases are symmetrical.
    By construction of our subdivision, we know that the cell has size at most \(\delta x\).
    We will now show that by dividing the cell one extra time (i.e. to a size of \(\delta x / 2\)), the furthest point will have the desired ratio irrespective of the value of \(y\).
    \begin{linenomath*}
    \begin{align*}
        \frac{\sqrt{(x + \frac{\delta x}{2})^2 + (y + \frac{\delta x }{2})^2}}{\sqrt{x^2 + y^2}}    & \leq 1 + \delta                                                       \\
        \frac{(x + \frac{\delta x}{2})^2 + (y + \frac{\delta x }{2})^2}{x^2 + y^2}                  & \leq (1 + \delta)^2                                                   \\
        \frac{x^2 + y^2 + \delta x^2 + \delta xy + \frac{\delta^2 x^2}{2}}{x^2 + y^2}               & \leq 1 + 2\delta + \delta^2                                           \\
        \frac{\delta x^2 + \delta xy + \frac{\delta^2 x^2}{2}}{x^2 + y^2}                           & \leq 2\delta + \delta^2                                               \\
        \delta xy                                                                                   & \leq \delta x^2 + \frac{\delta^2 x^2}{2} + 2\delta y^2 + \delta^2 y^2 \\
    \end{align*}
    \end{linenomath*}

    \noindent
    As \(\delta xy \leq \delta x^2\), and the other terms on the right-hand side are positive, the inequality holds.
    As such, the construction described can be turned into one where all cells have the desired ratio with one extra subdivision.

    Plugging the values of \(\alpha_i\), \(\beta\), \(\gamma\) into our initial formula, we can obtain the number of cells as a function of \(k\):
    \begin{linenomath*}
    \begin{align*}
            & \quad \left(\frac{\frac{\delta}{\sqrt{nm}}}{\frac{\delta^2}{\sqrt{nm}}}\right)^2 + \sum_{i = 0}^{k}{\frac{2}{\delta} \cdot \frac{\frac{\delta}{\sqrt{nm}} + \sum_{j = 0}^{i - 1}{\frac{1}{\delta}2^j\frac{\delta^2}{\sqrt{nm}}}}{2^i\frac{\delta^2}{\sqrt{nm}}} + \frac{1}{\delta^2}}  \\
        =   & \quad \frac{1}{\delta^2} + \frac{k}{\delta^2} + \frac{2}{\delta}\sum_{i = 0}^{k}{\frac{\frac{\delta}{\sqrt{nm}} + \frac{\delta}{\sqrt{nm}}\sum_{j = 0}^{i - 1}{2^j}}{2^i\frac{\delta^2}{\sqrt{nm}}}} \\
        =   & \quad \frac{1}{\delta^2} + \frac{k}{\delta^2} + \frac{2}{\delta}\sum_{i = 0}^{k}{\frac{\frac{\delta}{\sqrt{nm}} + \frac{\delta}{\sqrt{nm}}(2^i - 1)}{2^i\frac{\delta^2}{\sqrt{nm}}}} \\
        =   & \quad \frac{1}{\delta^2} + \frac{k}{\delta^2} + \frac{2}{\delta}\sum_{i = 0}^{k}{\frac{2^i\frac{\delta}{\sqrt{nm}}}{2^i\frac{\delta^2}{\sqrt{nm}}}} \\
        =   & \quad \frac{1}{\delta^2} + \frac{k}{\delta^2} + \frac{2}{\delta}\sum_{i = 0}^{k}{\frac{1}{\delta}} \\
        =   & \quad \frac{1}{\delta^2} + \frac{k}{\delta^2} + \frac{2k}{\delta^2} \\
        \in & \quad O\left(\frac{k}{\delta^2}\right)
    \end{align*}
    \end{linenomath*}

    We can directly calculate the value of \(k\) by considering the number of doublings needed to cover a horizontal line segment of length \(\Delta\) starting at \(p_i\):
    \begin{linenomath*}
    \begin{align*}
        \frac{\delta}{\sqrt{nm}} + \sum_{i = 0}^{k}{\frac{1}{\delta} \cdot 2^i \cdot \frac{\delta^2}{\sqrt{nm}}}    & = \Delta \\
        \frac{\delta}{\sqrt{nm}} + \frac{\delta}{\sqrt{nm}}\sum_{i = 0}^{k}{2^i}                                    & = \Delta \\
        1 + \sum_{i = 0}^{k}{2^i}                                                                                   & = \frac{\sqrt{nm}\Delta}{\delta} \\
        2^{k + 1}                                                                                                   & = \frac{\sqrt{nm}\Delta}{\delta} \\
        k                                                                                                           & \in O\left(\log\frac{nm\Delta}{\delta}\right)  \\
    \end{align*}
    \end{linenomath*}

    This gives a total number of cells of \(O\left(\frac{1}{\delta^2}\log\frac{nm\Delta}{\delta}\right)\) per point-triangle pair.
    Over all pairs, we obtain a total number of cells of \(O\left(\frac{nm}{\delta^2}\log\frac{nm\Delta}{\delta}\right)\).
    \hfill \qed
\end{proof}

\begin{lemma}\label{lem:point-triangle-count}
    The set \(Q\) has \(O\left(\dfrac{nm}{\delta^2}\log\dfrac{nm\Delta}{\delta}\right)\) cells.
\end{lemma}
\begin{proof}
    Consider any cell \(r \in R\).
    Any cell \(q \in Q\) that overlaps with \(r\) has \(\norm{q} \geq \norm{r}/16\): otherwise \(q\) was subdivided unnecessarily; see \cref{fig:point-quadtree-cells}.
    As the cells in \(Q\) are disjoint, it follows that \(r\) can overlap with at most 25 cells in \(Q\).
    As such, \(Q\) contains at most 25 times more cells than \(R\), which, by \cref{lem:point-triangle-alternative-count}, is \(O\left(\frac{nm}{\delta^2}\log\frac{nm\Delta}{\delta}\right)\).
    \hfill \qed
\end{proof}

This leads to the following result:

\begin{theorem}\label{thm:point-triangle-algorithm}
    Let \(P\) be a set of \(n\) weighted points and \(S\) be a set of \(m\) triangles with equal total weight, let \(\Delta\) be the longest edge length in \(S\) after normalising its total area to one, let \(\norm{\tp^*}\) be the cost of an optimal transport plan between \(P\) and \(S\), and let \(\delta\) be any constant \(> 0\).
    Given an algorithm that constructs a \((1 + \delta)\)-approximation between weighted sets of \(k\) points in \(f_\delta(k)\) time, we can construct a transport plan between \(P\) and \(S\) with cost \(\leq (1 + 9\delta)\norm{\tp^*}\) in \(O\left(f_\delta\left(\frac{nm}{\delta^2} \polylog\left(\frac{nm\Delta}{\delta}\right)\right)\right)\) time.
\end{theorem}

We can again calculate a \((1 + \delta)\)-approximation to \(\nu\) in \(O(N\delta^{-O(1)}\polylog{N})\) time using the algorithm by Fox and Lu~\cite{fox2022}, giving the following corollary to the previous theorem:

\begin{corollary}\label{cor:point-triangle-algorithm}
    For any constant \(\eps > 0\) and some constant \(c\), a transport plan between \(P\) and \(S\) with cost \(\leq (1 + \eps)\norm{\tp^*}\) can be constructed in \(O\left(\frac{nm}{\eps^c} \polylog\left(\frac{nm\Delta}{\eps}\right)\right)\) time with high probability.
\end{corollary}

\section{Segments to segments}\label{sec:segment-segment}
In the previous section, we considered the case when one of our two input sets consists of points.
We now describe an algorithm to compute the EMD between two sets of line segments.
Here, we cannot directly apply our general approach of subdividing: the optimal transport plan may have a cost arbitrarily close to zero.
As such, if we disregard everything within some radius of one of the sets, there may be nothing left.
We solve this by introducing an additive term into the approximation.
The cost of a plan generated by our algorithm is \((1 + \eps)\norm{\tp^*} + A\), for some value \(A\) depending on \(\eps\).
This allows us to greedily match parts of the input within a small distance of each other, and then solve the remainder with our previous approach.

Let \(P = \{p_1, \ldots, p_n\}\) and \(S = \{s_1, \ldots, s_m\}\) be sets of line segments with equal total length.
Our algorithm is then as follows.
First, we greedily match equal-length pieces of \(P\) and \(S\) that are within distance \(\delta / nm\) of each other, until no such pieces remain; we describe this process in more detail later.
Let \(P'\) and \(S'\) be the remaining parts of \(P\) and \(S\), respectively.
We subdivide \(P'\) and \(S'\) as before: for every \(p \in P'\), if there is an \(s \in S'\) such that the ratio between the closest and furthest distance is more than \(1 + \delta\), cut \(p\) in half; after processing \(P'\), do the same for \(S'\).
Call the resulting sets \(Q\) and \(R\).
We then choose a point on each \(q \in Q\) and \(r \in R\), with a weight equal to the length of the subsegment, and solve an optimal transport problem between these two point sets.
Our final transport plan is then obtained by spreading the mass moved between any two points evenly over the segments they were chosen on.

We first prove that greedily matching parts of the input within distance \(\delta / nm\) increases the cost of an optimal solution by at most an additive term.
The proof for the approximation algorithm then follows the same structure as in the previous sections.
Let \(\tp_M\) be a transport plan between the parts of the input that were greedily matched, in which the longest distance is at most \(\delta / nm\), and let \(\tp^*_G\) be an optimal transport plan for the remainder of the input.

\begin{figure}
    \centering
    \includegraphics{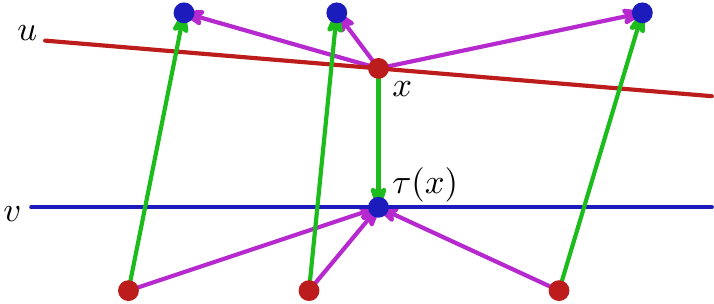}
    \caption{Two points, with their mass assignment in an optimal solution shown in purple.
    We greedily match a point \(x\) on \(u\) to \(\tau(x)\) on \(v\), and obtain the assignment of mass shown in green.}
    \label{fig:greedy-matching-swap}
\end{figure}

\begin{lemma}\label{lem:greedy-subsegment-matching}
    Let \(u\) and \(v\) be two subsegments with length \(l\) of \(P\) and \(S\), respectively.
    If all mass from \(u\) can be transported to \(v\) with distance at most \(\kappa\), then an optimal transport plan between \(P \setminus \{u\}\) and \(S \setminus \{v\}\) has cost at most \(\norm{\tp^*} + l\kappa\).
\end{lemma}
\begin{proof}
    We will construct a transport plan in which \(u\) and \(v\) are removed, having cost at most \(\norm{\tp^*} + l\kappa\).
    The cost of an optimal solution on the remainder is then not higher.

    Let \(\tp_s(x, t): \mathbb{R}^2 \times [0, 1] \rightarrow \mathbb{R}^2\) be a function describing, for a point \(x \in s\), where its mass comes from or goes to (recall that each point sends or receives mass density one), let \(d_s(x, t): \mathbb{R}^2 \times [0, 1] \rightarrow \mathbb{R}\) be defined as \(d(x, \tp_s(x, t))\), and let \(\tau: \mathbb{R}^2 \rightarrow \mathbb{R}^2\) be a mapping of points on \(u\) to points on \(v\) such that for all \(x \in u\), \(d(x, \tau(x)) \leq \kappa\).
    The cost of the part of \(\tp^*\) involving segments \(u\) and \(v\) can then be written as
    \begin{linenomath*}
    \begin{equation*}
        c^*(u) = \int_{x \in u}{\int{d_u(x, t)\dd t \dd x}}
    \end{equation*}
    \end{linenomath*}%
    \begin{linenomath*}
    \begin{equation*}
        c^*(v) = \int_{x \in u}{\int{d_v(\tau(x), t)\dd t \dd x}}
    \end{equation*}
    \end{linenomath*}
    We modify \(\tp^*\) by removing \(u\) and \(v\), and moving all mass that each point \(\tau(x)\) receives in \(\tp^*\) to where \(x\) moved it in \(\tp^*\); see \cref{fig:greedy-matching-swap}.
    We can distribute this mass in any way we like, as the total incoming and outgoing mass is one by definition.
    This gives a transport plan \(\tp\) with cost
    \begin{linenomath*}
    \begin{align*}
        \norm{\tp} = \norm{\tp^*}   & - \left(\int_{x \in u}{\int{d_u(x, t)\dd t \dd x}}\right) - \left(\int_{x \in u}{\int{d_v(\tau(x), t)\dd t \dd x}}\right) \\
                                    & + \left(\int_{x \in u}{\int{d(\tp_u(x, t), \tp_v(\tau(x), t))\dd t \dd x}}\right)
    \end{align*}
    \end{linenomath*}

    \noindent
    By the triangle inequality, \(d(\tp_u(x, t), \tp_v(\tau(x), t)) \leq d_v(\tau(x), t) + \kappa + d_u(x, t)\).
    It follows that
%
\begin{align}
        & \norm{\tp}    && \leq \norm{\tp^*}    && - \left(\int_{x \in u}{\int{d_u(x, t)\dd t \dd x}}\right) - \left(\int_{x \in u}{\int{d_v(\tau(x), t)\dd t \dd x}}\right) \\
        &               &&                      && + \left(\int_{x \in u}{\int{d_v(\tau(x), t) + \kappa + d_u(x, t)\dd t \dd x}}\right)                                                           \\
        &               && = \norm{\tp^*}       && + l\kappa \tag*{}
\end{align}
\mbox{ }\hfill\qed
\end{proof}

Note that this bound is tight in the worst case: consider horizontal line segments of unit length with their left endpoints having \(x\)-coordinate 0. If we take \(P\) to consist of two such segments at \(y = 0\) (\(p_1\)) and \(y = 2\) (\(p_2\)), and \(S\) to consist of two segments at \(y = 1\) (\(s_1\)) and \(y = 3\) (\(s_2\)), the optimal solution would move mass from \(p_1\) to \(s_1\) and from \(p_2\) to \(s_2\), giving a total cost of \(2\). If we set \(\kappa = 1\), we would greedily match \(p_2\) and \(s_1\), giving a total cost of \(3\), being exactly \(l\kappa\) more than the optimal.

We can now bound the costs of \(\tp^*_G\) and \(\tp_M\).

\begin{lemma}\label{lem:greedy-matching-remainder-upper-bound}
    \(\norm{\tp^*_G} \leq \norm{\tp^*} + \dfrac{\delta}{nm}\).
\end{lemma}
\begin{proof}
    Any subsegments of \(P\) and \(S\) with length \(l\) that are greedily matched increase the cost of an optimal solution in the remaining part by at most \(\delta l / nm\) (\cref{lem:greedy-subsegment-matching}).
    The total length that can be greedily matched is at most one, so the total extra cost is at most \(\delta / nm\).
    \hfill \qed
\end{proof}

\begin{lemma}\label{lem:greedy-matching-upper-bound}
    \(\norm{\tp_M} \leq \dfrac{\delta}{nm}\).
\end{lemma}
\begin{proof}
    By construction, the distance over which any mass is transported in \(\tp_M\) is at most \(\delta / nm\).
    The total mass transported is at most one, giving the bound.
    \hfill \qed
\end{proof}

For each segment \(p \in P\), we can straightforwardly compute a maximal subset that can be transported over distance at most \(\delta / nm\).
Consider each segment \(s \in S\): the supporting lines of \(p\) and \(s\) can intersect inside \(p\) or \(s\), outside both, or not at all.
If they don't intersect (i.e.
are parallel), computation of the parts that can be transported within the required distance is trivial.
If they intersect outside both, we can find the points on \(p\) and \(s\) furthest from the intersection point that are within the required distance, then find the largest distance we can move towards the intersection point while staying within the required distance.
If they intersect inside one or both of the segments, we split the segments at the intersection point and handle both sides using the case for intersections outside the segments.

Let \(P'\) and \(S'\) be the parts of \(P\) and \(S\) that remain after the greedy matching, with \(\norm{P'} = \norm{S'} = \ell\).
We subdivide \(P'\) and \(S'\) into \(Q\) and \(R\) as described above.
As before, we define a complete bipartite graph \(G = (Q \cup R, Q \times R)\), where the weight of each edge is equal to the shortest distance between the two subsegments it connects, and the capacity of each vertex is equal to the length of the subsegment it represents.
Let \(\mathcal{W}\) be a minimum cost flow in \(G\); we observe the following lower bound on its cost:

\begin{lemma}\label{lem:lower-bound-subdivision-segment-segment}
    \(\norm{\mathcal{W}} \geq \dfrac{\delta\ell}{nm}\).
\end{lemma}
\begin{proof}
    By construction, the distances in \(G\) are at least \(\delta / nm\).
    As the total mass moved is \(\ell\), we obtain the bound stated in the lemma.
    \hfill \qed
\end{proof}

\cref{lem:lower-bound-flow} also still applies to the part of the input that remains after greedy matching.
We now approximate \(\norm{\mathcal{W}}\) by reducing the flow problem to a transportation problem between two weighted point sets.
We pick any point on each \(q \in Q\) and \(r \in R\), and give them weights equal to \(\norm{q}\) and \(\norm{r}\).
Call these sets of points \(U\) and \(V\).
We can now bound the cost of \(\tp_G\) in terms of \(\tp^*_G\) using the flow problem.

\begin{lemma}\label{lem:segment-segment-bounds}
    \(\norm{\tp^*_G} \leq \norm{\tp_G} \leq (1 + \delta)^2\norm{\tp^*_G}\).
\end{lemma}
\begin{proof}
    Let \(\nu\) be an optimal transport plan between \(U\) and \(V\), and let \(\norm{\nu}\) be its cost.
    We can upper bound \(\norm{\nu}\) by measuring all distances to the furthest points inside the segments.
    By construction of \(Q\) and \(R\), the ratio of longest to shortest distance is at most \(1 + \delta\).
    The cost \(\norm{\nu}\) of \(\nu\) can therefore not be more than \((1 + \delta)\norm{\mathcal{W}}\).

    We can turn \(\tp\) into a valid transport plan \(\tp_G\) between \(Q\) and \(R\) by spreading the mass moved to each point in \(U\) and \(V\) evenly over the segments in \(Q\) and \(R\) that they were derived from.
    Again, by construction, the distances increase by a factor of at most \(1 + \delta\), giving \(\tp_G \leq (1 + \delta)\norm{\nu}\).

    Plugging in the bound on \(\norm{\nu}\) obtained above, we obtain an upper bound of \((1 + \delta)^2\norm{\mathcal{W}}\).
    As \(\norm{\mathcal{W}} \leq \norm{\tp^*_G}\), we obtain that \(\norm{\tp_G} \leq (1 + \delta)^2\norm{\tp^*_G}\).
    The lower bound follows directly from the fact that \(\tp^*_G\) is optimal, and therefore cannot have a cost higher than that of \(\tp_G\).
    \hfill \qed
\end{proof}

We can then show that, for a transport plan \(\tp = \tp_G + \tp_M\), \(\norm{\tp}\) approximates \(\norm{\tp^*}\):

\begin{theorem}\label{thm:segment-segment-approximation}
    \(\norm{\tp} \leq (1 + 3\delta)\norm{\tp^*} + \dfrac{5\delta}{nm}\).
\end{theorem}
\begin{proof}
    By \cref{lem:segment-segment-bounds}, we know that \(\norm{\tp_G} \leq (1 + \delta)^2\norm{\tp^*_G}\).
    As \(\delta \leq 1\), \((1 + \delta)^2 \leq 1 + 3\delta\).
    By \cref{lem:greedy-matching-remainder-upper-bound}, we have that \(\norm{\tp^*_G} \leq \norm{\tp^*} + (1 - \ell)\delta / nm\).
    Combining the two results, we get that
    \begin{linenomath*}
    \begin{align*}
        \norm{\tp_G}    & \leq (1 + 3\delta)\norm{\tp^*_G}                                          \\
                        & \leq (1 + 3\delta)\left(\norm{\tp^*} + (1 - \ell)\frac{\delta}{nm}\right) \\
                        & \leq (1 + 3\delta)\norm{\tp^*} + (1 - \ell)\frac{\delta + 3\delta^2}{nm}  \\
                        & \leq (1 + 3\delta)\norm{\tp^*} + \frac{4\delta}{nm}                       \\
    \end{align*}
    \end{linenomath*}

    By \cref{lem:greedy-matching-upper-bound}, \(\norm{\tp_M} \leq \delta / nm\).
    As \(\norm{\tp} = \norm{\tp_G} + \norm{\tp_M}\), we obtain the bound stated in the lemma.
    \hfill \qed
\end{proof}

Setting \(\delta = \eps / 3\) gives a \((1 + \eps)\)-approximation with an additive term of \(5\eps / 3nm\)

\subsection{Running time analysis}\label{sec:segment-segment-analysis}
During the greedy matching, each \(p \in P\) may have been cut into \(m\) pieces, and each \(s \in S\) into \(n\) pieces.
As such, \(P'\) and \(S'\) (the parts remaining after greedy matching) both contain \(O(nm)\) subsegments.
In the worst case, \(P'\) and \(S'\) are close to each other everywhere, causing them to be subdivided into the smallest possible subsegments.
As the minimum distance is \(\delta / nm\), and the ratio of the longest and shortest distance between any two subsegments is \(1 + \delta\), the smallest possible subsegment has size \(\Theta(\frac{\delta^2}{nm})\).
Each subsegment of \(P'\) and \(S'\) may give rise to one extra subsegment in \(Q\) and \(R\), as the length may not be exactly divisible by \(\delta / nm\).
This gives sets \(Q\) and \(R\) a size of \(O(\frac{nm}{\delta^2} + nm) \in O(\frac{nm}{\eps^2})\), leading to the following result:

\begin{theorem}\label{thm:segment-segment-algorithm}
    Let \(P\) and \(S\) be sets of \(n\) and \(m\) line segments in the plane, both having equal total length, let \(\norm{\tp^*}\) be the cost of an optimal transport plan between them, and let \(\delta\) be any constant \(> 0\).
    Given an algorithm that constructs a \((1 + \delta)\)-approximation between weighted sets of \(k\) points in \(f_\delta(k)\) time, we can construct a transport plan between \(P\) and \(S\) with cost \(\leq (1 + c'\delta)\norm{\tp^*} + \frac{5\delta}{nm}\) for some constant \(c'\) in \(O\left(f_\delta\left(\frac{nm}{\delta^2}\right)\right)\) time with high probability.
\end{theorem}

We can again calculate a \((1 + \delta)\)-approximation to \(\nu\) in \(O(N\delta^{-O(1)}\polylog{N})\) time using the algorithm by Fox and Lu~\cite{fox2022}, giving the following corollary to the previous theorem:

\begin{corollary}\label{cor:segment-segment-algorithm}
    For any constant \(\eps > 0\), a transport plan between \(P\) and \(S\) with cost \(\leq (1 + \eps)\norm{\tp^*} + O\left(\frac{\eps}{nm}\right)\) can be constructed in \(O\left(\frac{nm}{\eps^c}\polylog\left(\frac{nm}{\eps}\right)\right)\) time with high probability.
\end{corollary}

\section{Triangles to triangles}\label{sec:triangle-triangle}
We consider the case where \(P\) and \(S\) are both sets of triangles with total area one and longest edge length \(\Delta\).
The algorithm is completely analogous to the one for transport between sets of segments: we greedily match parts of the input within a certain distance, subdivide the remainder and approximate the optimal transport plan by reduction to a minimum cost flow.
As the setup and proofs are exactly the same as in the previous section (just substitute the integrals over segments with integrals over area), this is omitted.
All we need is an algorithm that can greedily match parts of the input within a given distance.

We do this greedy matching as follows.
We can first remove the parts where \(P\) and \(S\) overlap: they have cost zero.
We then overlay a grid with cells of size \(\delta / (2\sqrt{nm})\) onto our input, and keep only the cells that contain an edge or are adjacent to one that does (the other cells already have the desired clearance from cells containing triangles from the other set).
Inside each cell, we record the total area of triangles from \(P\) and \(S\) that lie inside it separately.
For parts of \(P\) and \(S\) that lie inside the same cell, we match as much as possible, resulting in a grid where each cell only contains parts of \(P\) or \(S\).
We then match as much of each cell as possible to each of its eight neighbours.
The maximum distance over which we have greedily matched weight is \(\sqrt{2}\delta / \sqrt{nm}\), and the remaining parts of \(P\) and \(S\) have a minimum distance of \(\delta / \sqrt{nm}\) to each other.

We can then combine the ideas of the point-to-triangle and segment-to-segment algorithm to approximate the optimal solution.
First, we overlay a uniform grid with cells of size \(\delta \times \delta\) separately for \(P\) and \(S\), and identify the cells that contain a triangle.
We then recursively subdivide each cell as long as there is any part of a triangle in the other set for which our distance ratio of \(1 + \delta\) is violated.
We track the total area of triangles contained in each cell, then approximate each cell by a point with a weight equal to this area.
After each set is approximated by points in this way, we can run our algorithm as before.
\cref{lem:greedy-subsegment-matching} can be straightforwardly modified to give the same result for triangles, except that \(l\) will be an area instead of a length.
This means that we obtain a \((1 + \eps)\)-approximation with an additive term of \(O(\eps / \sqrt{nm})\).

\subsection{Running time analysis}\label{sec:triangle-triangle-analysis}
The number of cells examined during the greedy matching is \(O(\frac{\sqrt{nm}\Delta}{\delta})\) per triangle, so \(O(\frac{\sqrt{nm}\Delta(n + m)}{\delta})\) in total.
The part of the input remaining after greedy matching can be at distance \(\delta / \sqrt{nm}\) from each other everywhere, causing it to be subdivided into cells of size \(\Theta(\frac{\delta^2}{\sqrt{nm}})\) to maintain a distance ratio of \(1 + \delta\).
There may be \(O(\frac{\sqrt{nm}\Delta}{\delta^2})\) cells that intersect the boundaries of a triangle, or \(O(\frac{\sqrt{nm}\Delta(n + m)}{\delta^2})\) in total; the other cells are interior to some triangle, and as the total area is one, there can be at most \(O(\frac{nm}{\delta^4})\) of them.
The total number of cells is therefore at most \(O(\frac{\sqrt{nm}\Delta(n+m)}{\delta^2} + \frac{nm}{\delta^4}) \in O(\frac{\sqrt{nm}\Delta(n+m)}{\delta^4}\).
This gives the following result:

\begin{theorem}\label{thm:triangle-triangle-algorithm}
    Let \(P\) be a set of \(n\) and \(S\) a set of \(m\) triangles in the plane, both having equal total area and longest edge length at most \(\Delta\) after normalising their total areas to one, let \(\norm{\tp^*}\) be the cost of an optimal transport plan, and let \(\delta\) be any constant \(> 0\).
    Given an algorithm that constructs a \((1 + \delta)\)-approximation between weighted sets of \(k\) points in \(f_\delta(k)\) time, we can construct a transport plan between \(P\) and \(S\) with cost \(\leq (1 + c'\delta)\norm{\tp^*} + O(\frac{\delta}{\sqrt{nm}})\) for some constant \(c'\) can be constructed in \(O\left(f_\delta\left(\frac{\sqrt{nm}\Delta(n + m)}{\delta^4}\right)\right)\) time.
\end{theorem}

We can again calculate a \((1 + \delta)\)-approximation to \(\nu\) in \(O(N\delta^{-O(1)}\polylog{N})\) time using the algorithm by Fox and Lu~\cite{fox2022}, giving the following corollary to the previous theorem:

\begin{corollary}\label{cor:triangle-triangle-algorithm}
    For any constant \(\eps > 0\), a transport plan between \(P\) and \(S\) with cost \(\leq (1 + \eps)\norm{\tp^*} + O(\frac{\eps}{\sqrt{nm}})\) can be constructed in \(O\left(\frac{\sqrt{nm}\Delta(n + m)}{\eps^c}\polylog\left(\frac{nm\Delta}{\eps}\right)\right)\) time with high probability.
\end{corollary}

\section{Higher dimensions}\label{sec:higher-dimensions}
In this section we show how our approach can be extended to work in \(d\)-dimensional space.
We discuss the case of transporting mass from points to \(d\)-dimensional simplices, and from one set of simplices to another.

\subsection{Points to simplices}\label{sec:point-simplex}
The approach described here is a direct extension of the one detailed in \cref{sec:point-triangle}.
Let \(P\) be a set of \(n\) weighted points in \(d\) dimensions with total mass one, and let \(S\) be a set of \(m\) \(d\)-dimensional simplices with total volume one and longest edge length \(\Delta\).
We start by overlaying an infinite grid of size \(\Delta\) and identifying the cells intersected by any simplex in \(O(dm\log(m) + 6^dd^2d^dm)\) time using compressed quadtrees~\cite{vanderhoog2018}.
We then repeatedly subdivide each cell until the ratio between the shortest and longest distance is at most \(1 + \delta\) for all points in \(P\), or until it is wholly within \(\delta / (nm)^{1/d}\) of any point in \(P\).
Call the resulting set of cells \(Q\).
We can again show that picking one point in each cell of \(Q\) with weight equal to the total volume of simplices contained in it, and then solving a transport problem between \(P\) and the resulting set of points, approximates the transport problem between \(P\) and \(S\).

As the structure of the proof is very similar to that contained in \cref{sec:point-triangle}, we omit some of the intermediate lemmas here.
We start with the lower bound on the cost of a minimum cost flow \(\mathcal{W}\) in the bipartite graph \(G = (P \cup Q, P \times Q)\):
\begin{lemma}\label{lem:lower-bound-subdivision-point-simplex}
    \(\norm{\mathcal{W}} \geq \dfrac{\delta}{(nm)^{1/d}} - \dfrac{\delta(2(\delta + \delta^2))^d}{(nm)^{1/d}}\).
\end{lemma}
\begin{proof}
    For a given point-simplex pair \((p, s) \in P \times S\), consider the cells in \(Q\) intersecting \(s\) that have a point within distance \(\delta / (nm)^{1/d}\) of \(p\).
    Such a cell has its furthest point at most at distance \((\delta + \delta^2) / (nm)^{1/d}\).
    The total volume of these cells is then \((2(\delta + \delta^2))^d / nm\) (the volume of a \(d\)-dimensional hypercube with radius \((\delta + \delta^2) / (nm)^{1/d}\), which contains the hypersphere with the same radius).
    Over all point-simplex pairs, this gives a volume of at most \((2(\delta + \delta^2))^d\), leaving \(1 - (2(\delta + \delta^2))^d\) with distance at least \(\delta / (nm)^{1/d}\) in \(\mathcal{W}\).
    The cost is therefore at least \((1 - (2(\delta + \delta^2))^d) \cdot \delta / (nm)^{1/d} = \delta / (nm)^{1/d} - \delta(2(\delta + \delta^2))^d / (nm)^{1/d}\).
    \hfill \qed
\end{proof}
We again approximate \(\norm{\mathcal{W}}\) by reducing the flow problem to a transportation problem between two sets of weighted points.
We do this by picking a point in each cell \(q \in Q\) and giving it a weight equal to the total volume of simplices contained in \(q\).
Call this set of points \(T\):
\begin{lemma}\label{lem:approximate-flow-point-simplex}
    The cost \(\norm{\nu}\) of an optimal transport plan \(\nu\) between \(P\) and \(T\) satisfies \(\norm{\mathcal{W}} \leq \norm{\nu} \leq (1 + \delta)\norm{\mathcal{W}} + \dfrac{\delta(2(\delta + \delta^2))^d}{(nm)^{1/d}}\).
\end{lemma}
\begin{proof}
    The lower bound follows directly from the fact that none of the distances in \(\nu\) are smaller than the distances between the same objects in \(\mathcal{W}\) (recall that the distances in \(\mathcal{W}\) are measured to the closest point on the cell).
    We can upper bound \(\norm{\nu}\) by measuring all distances to the furthest point in each cell.
    By construction, those distances are at most \(1 + \delta\) times the distance to the closest point when the furthest distance is at least \(\delta / (nm)^{1/d}\).
    We can therefore bound all parts of \(\nu\) where the distance is at least \(\delta / (nm)^{1/d}\) by \((1 + \delta)\norm{\mathcal{W}}\).
    The total mass being moved over distance at most \(\delta / (nm)^{1/d}\) in \(\nu\) is at most \((2(\delta + \delta^2))^d\), giving a cost of \(\delta(2(\delta + \delta^2))^d / (nm)^{1/d}\).
    The total cost when measuring to the furthest point is therefore \((1 + \delta)\norm{\mathcal{W}} + \delta(2(\delta + \delta^2))^d / (nm)^{1/d}\).
    \hfill \qed
\end{proof}
As before, we turn \(\nu\) into a valid transport plan \(\tp\) between \(P\) and \(Q\) by spreading the mass moved to each point in \(T\) out evenly over the parts of simplices contained in the cell of \(Q\) the point was derived from.
By the same argument used in \cref{lem:approximate-flow-point-simplex}, we obtain the following bound on the cost of \(\tp\):
\begin{lemma}\label{lem:upper-bound-discretised-point-simplex}
    \(\norm{\tp} \leq (1 + \delta)\norm{\nu} + \dfrac{\delta(2(\delta + \delta^2))^d}{(nm)^{1/d}}\).
\end{lemma}

Putting this all together, we can show that \(\norm{\tp}\) approximates the cost of the optimal transport plan \(\norm{\tp^*}\):
\begin{theorem}\label{thm:point-simplex-approximation}
    \(\norm{\tp}\) is a \((1 + 21\delta)\)-approximation to the earth mover's distance \(\norm{\tp^*}\) between \(P\) and \(S\) for \(0 < \delta \leq \frac{1}{5}\).
\end{theorem}
\begin{proof}
    We follow the same structure as \cref{thm:point-triangle-approximation}, obtaining the following ratio, into which we plug the lower bound from \cref{lem:lower-bound-subdivision-point-simplex}:
    \begin{linenomath*}
    \begin{align*}
                & \quad \frac{(1 + \delta)^2\norm{\mathcal{W}} + \frac{2\delta(2(\delta + \delta^2))^d}{(nm)^{1/d}} + \frac{\delta^2(2(\delta + \delta^2))^d}{(nm)^{1/d}}}{\norm{\mathcal{W}}} \\
        \leq    & \quad \frac{(1 + 2\delta + \delta^2)\left(\frac{\delta}{(nm)^{1/d}} - \frac{\delta(2(\delta + \delta^2))^d}{(nm)^{1/d}}\right) + \frac{2\delta(2(\delta + \delta^2))^d}{(nm)^{1/d}} + \frac{\delta^2(2(\delta + \delta^2))^d}{(nm)^{1/d}}}{\frac{\delta}{(nm)^{1/d}} - \frac{\delta(2(\delta + \delta^2))^d}{(nm)^{1/d}}} \\
        =       & \quad \frac{1 + 2\delta + \delta^2 + (2(\delta + \delta^2))^d - \delta(2(\delta + \delta^2))^d - \delta^2(2(\delta + \delta^2))^d}{1 - (2(\delta + \delta^2))^d} \\
        =       & \quad 1 + \delta + \frac{\delta + \delta^2 + 2(2(\delta + \delta^2))^d - \delta^2(2(\delta + \delta^2))^d}{1 - (2(\delta + \delta^2))^d} \\
        \leq    & \quad 1 + \delta + \frac{\delta + \delta^2 + 2(2(\delta + \delta^2))((2(\delta + \delta^2)))^{d-1}}{1 - (2(\delta + \delta^2))^d} \\
        \leq    & \quad 1 + \delta + \frac{\delta + \delta^2 + 8\delta}{1 - \frac{1}{2^d}} \tag{Assuming \((\delta + \delta^2) \leq \frac{1}{4}\)} \\
        \leq    & \quad 1 + \delta + \frac{9\delta + \delta^2}{\frac{1}{2}} \\
        <       & \quad 1 + 21\delta \\
    \end{align*}
    \end{linenomath*}

    As \(\norm{\mathcal{W}}\) is also a lower bound for \(\norm{\tp}\) (\cref{lem:lower-bound-flow}), and \(\tp\) can obviously not have lower cost than the optimal transport plan, this gives a \((1 + 21\delta)\)-approximation.
    \hfill \qed
\end{proof}

Note that the constant in our approximation is slightly worse than the one obtained in \cref{thm:point-triangle-approximation}; this is because we approximate the volume of a hypersphere by the volume of its bounding cube, whereas before we could calculate the area of the disk exactly.

\subsubsection{Running time analysis}\label{sec:point-simplex-analysis}
We construct a structure similar to \cref{sec:point-triangle-analysis}, then argue that the number of cells in our actual subdivision is similar.
Our construction is the direct generalisation of the one described before: we build a layered structure of cells of increasing sizes.
Let \(R_{i,j}\) be the set of cells generated by point \(p_i\) and simplex \(s_j\); we now analyse how many cells are created.

\begin{lemma}\label{lem:point-simplex-alternative-count}
    \(R = \bigcup R_{i,j}\) has \(O\left(\dfrac{dd^{d/2}nm}{\delta^d}\log\dfrac{(nm)^{1/d}\Delta}{\delta}\right)\) cells.
\end{lemma}
\begin{proof}
    We follow the structure of the proof of \cref{lem:point-triangle-count}, again counting the number of cells in one ``quadrant'', and then multiplying by the number of quadrants (\(2^d\)).
    The number of cells in a quadrant is
    \begin{linenomath*}
    \begin{equation*}
        \left(\frac{\beta}{\gamma}\right)^d + \sum_{i = 0}^{k}{d\alpha_i \cdot \left(\frac{\beta + \sum_{j = 0}^{i - 1}{\alpha_j 2^j \gamma}}{2^i \gamma}\right)^{d - 1} + \alpha_i^d}
    \end{equation*}
    \end{linenomath*}
    where \(\alpha_i\) is the number of layers of cells that have doubled in size \(i\) times, \(\beta\) is the distance inside of which we use the smallest cell size (\(\delta / (nm)^{1/d}\)), \(\gamma\) is the smallest cell size (\(\delta^2 / (nm)^{1/d}\)), and \(k\) is the number of times we need to double the cell size.

    The value of each \(\alpha_i\) is \(1 / \delta\): this value is exact along any axis, and we can show that all cells can be made to have the correct ratio with a given number of extra subdivisions.
    For a given cell \(r \in R_{i,j}\), let \(\mathbf{v}\) be the vector from \(p_i\) to the closest point on \(r\), and let \(\gamma'\) be the edge length of cell \(q\).
    W.l.o.g. assume that \(v_0 = \max v_i\); by construction this gives us that \((v_0 + \gamma') / v_0 \leq 1 + \delta\).
    Let \(\mathbf{u}\) be the vector from the closest point on \(q\) to the furthest point; we want to find a value \(x\) such that \(\frac{\norm{\mathbf{v} + \mathbf{u} / x}}{\norm{\mathbf{v}}} \leq 1 + \delta\).
    Through the triangle inequality, we can upper bound the distance to the furthest point on \(q\) as \(\norm{\mathbf{v}} + \norm{\mathbf{u} / x}\).
    We can now calculate the required value of \(x\):
    \begin{linenomath*}
    \begin{align*}
        \frac{\norm{\mathbf{v}} + \norm{\frac{\mathbf{u}}{x}}}{\norm{\mathbf{v}}}   & \leq 1 + \delta \\
        \frac{\norm{\mathbf{u}}}{x\norm{\mathbf{v}}}                                & \leq \delta \\
        \frac{\gamma'\sqrt{d}}{x\norm{\mathbf{v}}}                                  & \leq \delta \\
        \frac{\gamma'\sqrt{d}}{xv_0}                                                & \leq \delta \\
        \frac{\sqrt{d}}{x}\delta                                                    & \leq \delta \\
        x                                                                           & \geq \sqrt{d}
    \end{align*}
    \end{linenomath*}
    So all cells have the correct ratio if their edge length is reduced by a factor of at least \(\sqrt{d}\), which means each cell needs to be replaced by at most \(O(d^{d/2})\) cells.

    This gives us the following derivation for the number of cells:
    \begin{linenomath*}
    \begin{align*}
            & \quad \frac{1}{\delta^d} + \sum_{i = 0}^{k}{\frac{d}{\delta} \cdot \left(\frac{\frac{\delta}{(nm)^{1/d}} + \sum_{j = 0}^{i - 1}{2^j \frac{\delta}{(nm)^{1/d}}}}{2^i \frac{\delta^2}{(nm)^{1/d}}}\right)^{d-1} + \frac{1}{\delta^d}} \\
        =   & \quad \frac{1}{\delta^d} + \frac{k}{\delta^d} + \frac{d}{\delta}\sum_{i = 0}^{k}{\left(\frac{\frac{\delta}{(nm)^{1/d}} + \frac{\delta}{(nm)^{1/d}}(2^i - 1)}{2^i \frac{\delta^2}{(nm)^{1/d}}}\right)^{d-1}} \\
        =   & \quad \frac{1}{\delta^d} + \frac{k}{\delta^d} + \frac{d}{\delta}\sum_{i = 0}^{k}{\left(\frac{\frac{\delta}{(nm)^{1/d}}}{\frac{\delta^2}{(nm)^{1/d}}}\right)^{d-1}} \\
        =   & \quad \frac{1}{\delta^d} + \frac{k}{\delta^d} + \frac{d}{\delta} \cdot \frac{k}{\delta^{d-1}} \\
        =   & \quad \frac{1}{\delta^d} + \frac{k}{\delta^d} + \frac{kd}{\delta^d} \\
        \in & \quad O\left(\frac{kd}{\delta^d}\right)
    \end{align*}
    \end{linenomath*}

    The value of \(k\) is derived in the same way as before, giving \(k \in O(\log((nm)^{1/d}\Delta / \delta))\).
    This gives \(O(\frac{d}{\delta^d}\log\frac{(nm)^{1/d}\Delta}{\delta})\) cells per point-simplex pair, where each cell needs to be divided into \(O(d^{d/2})\) smaller cells, for a total of \(O(\frac{dd^{d/2}nm}{\delta^d}\log\frac{(nm)^{1/d}\Delta}{\delta})\) cells.
    \hfill \qed
\end{proof}

\begin{lemma}\label{lem:point-simplex-count}
    \(Q\) has \(O\left(\dfrac{5^ddd^{d/2}nm}{\delta^d}\log\dfrac{(nm)^{1/d}\Delta}{\delta}\right)\) cells.
\end{lemma}
\begin{proof}
    Consider any cell \(r \in R\).
    As before, any cell \(q \in Q\) that overlaps with \(r\) has \(\norm{q} \geq \norm{r}/4^d\): otherwise \(q\) was subdivided unnecessarily.
    As the cells in \(Q\) are disjoint, it follows that \(r\) can overlap with at most \(5^d\) cells in \(Q\).
    As such, \(Q\) contains at most \(5^d\) times more cells than \(R\), which, by \cref{lem:point-simplex-alternative-count}, is \(O\left(\frac{dd^{d/2}nm}{\delta^d}\log\frac{(nm)^{1/d}\Delta}{\delta}\right)\).
    \hfill \qed
\end{proof}

Combined with the time required to build the quadtree that we use to find the starting cells of our subdivision, this gives the following result:

\begin{theorem}\label{thm:point-simplex-algorithm}
    Let \(P\) be a set of \(n\) weighted points and \(S\) be a set of \(m\) simplices in \(\mathbb{R}^d\) with equal total weight, let \(\Delta\) be the longest edge length in \(S\) after normalising its total volume to one, let \(\norm{\tp^*}\) be the cost of an optimal transport plan between \(P\) and \(S\), and let \(\delta\) be any constant \(> 0\).
    Given an algorithm that constructs a \((1 + \delta)\)-approximation between weighted sets of \(k\) points in \(f_\delta(k)\) time, we can construct a transport plan between \(P\) and \(S\) with cost \(\leq (1 + 21\delta)\norm{\tp^*}\) in \(O\left(6^dd^2d^dm + f_\delta\left(\frac{5^ddd^{d/2}nm}{\delta^d}\log\frac{(nm)^{1/d}\Delta}{\delta}\right)\right)\) time.
\end{theorem}

Note that if we set \(d = 2\), this is the same running time as the one obtained in \cref{thm:point-triangle-algorithm}.
We can again calculate a \((1 + \delta)\)-approximation to \(\nu\) in \(O(N\delta^{-O(d)}\log^{O(d)}{N})\) time using the algorithm by Fox and Lu~\cite{fox2022}.
Setting \(\delta = \eps/21\), this gives a total running time of
\begin{linenomath*}
\begin{align*}
    & \quad O\left(6^dd^2d^dm + \frac{105^ddd^{d/2}nm}{\eps^{O(d)}} \log^{O(d)}\left(\frac{5^ddnm\Delta}{\eps^d}\right)\right) \\
    \in & \quad O\left(6^dd^2d^dm + \frac{105^dd^2d^{d/2}nm}{\eps^{O(d)}} \log^{O(d)}\left(\frac{dnm\Delta}{\eps^d}\right)\right) \text{.}
\end{align*}
\end{linenomath*}

\begin{corollary}\label{cor:point-simplex-algorithm}
    For any constant \(\eps > 0\), a transport plan between \(P\) and \(S\) with cost \(\leq (1 + \eps)\norm{\tp^*}\) can be constructed in \(O\left(6^dd^2d^dm + \frac{105^dd^2d^{d/2}nm}{\eps^{O(d)}} \log^{O(d)}\left(\frac{dnm\Delta}{\eps^d}\right)\right)\) time with high probability.
\end{corollary}

\subsection{Simplices to simplices}\label{sec:simplex-simplex}
The approach from \cref{sec:segment-segment,sec:triangle-triangle} can also be extended to work on \(d\)-dimensional simplices in \(d\) dimensions.
We take the same approach of overlaying a grid with cells of size \(\delta/(4(nm)^{1/d})\) onto the input and greedily matching the parts of \(P\) and \(S\) that are close together.
The maximum distance over which we greedily match weight is then \(\sqrt{d}\delta/(2(nm)^{1/d})\), and the remaining parts of \(P\) and \(S\) have minimum distance \(\delta/(2(nm)^{1/d})\) to each other.
We then approximate the transport plan between the remaining cells with a minimum cost flow.
The same analysis still works, and we obtain a \((1 + \eps)\)-approximation with an additive term of \(O(\sqrt{d}\eps/(nm)^{1/d})\).

\subsubsection{Running time analysis}\label{sec:simplex-simplex-analysis}
The number of cells examined during the greedy matching is \(O(\frac{(nm)^{1/d}\Delta^d}{\delta})\) per simplex, so \(O(\frac{(nm)^{1/d}\Delta^d(n+m)}{\delta})\) in total (note that we simplify the analysis by simply considering the volume of a \(d\)-dimensional cube of side length \(\Delta\)).
The part of the input remaining after greedy matching can be close to each other everywhere, causing it to be subdivided into cells of size \(\Theta(\frac{\delta^2}{\sqrt{d}(nm)^{1/d}})\).
The total number number of these cells is then \(O(\frac{\sqrt{d}(nm)^{1/d}\Delta^d(n+m)}{\delta^2})\).
This gives the following result:

\begin{theorem}\label{thm:simplex-simplex-algorithm}
    Let \(P\) be a set of \(n\) and \(S\) a set of \(m\) \(d\)-dimensional simplices in \(\mathbb{R}^d\), both having equal total volume and longest edge length at most \(\Delta\) after normalising their total volumes to one, let \(\norm{\tp^*}\) be the cost of an optimal transport plan between them, and let \(\delta\) be any constant \(> 0\).
    Given an algorithm that constructs a \((1 + \delta)\)-approximation between weighted sets of \(k\) points in \(f_\delta(k)\) time, we can construct a transport plan between \(P\) and \(S\) with cost \(\leq (1 + c'\delta)\norm{\tp^*} + O(\frac{\sqrt{d}\delta}{(nm)^{1/d}})\) for some constant \(c'\) can be constructed in \(O\left(f_\delta\left(\frac{\sqrt{d}(nm)^{1/d}\Delta^d(n+m)}{\delta^2}\right)\right)\) time.
\end{theorem}

We can again calculate a \((1 + \delta)\)-approximation to \(\nu\) in \(O(N\delta^{-O(d)}\log^{O(d)}{N})\) time using the algorithm by Fox and Lu~\cite{fox2022}, giving the following corollary to the previous theorem:

\begin{corollary}\label{cor:simplex-simplex-algorithm}
    For any constant \(\eps > 0\), a transport plan between \(P\) and \(S\) with cost \(\leq (1 + \eps)\norm{\tp^*} + O(\frac{\sqrt{d}\eps}{(nm)^{1/d}})\) can be found in \(O\left(\frac{\sqrt{d}(nm)^{1/d}\Delta^d(n+m)}{\eps^{O(d)}}\log^{O(d)}\left(\frac{d(nm)^{1/d}\Delta^d}{\eps}\right)\right)\) time with high probability.
\end{corollary}

\section{Conclusion}
We have provided approximation algorithms to the earth mover's distance between sets of points, line segments, triangles and \(d\)-dimensional simplices.
These are the first combinatorial algorithms with a provable approximation ratio for this problem when the objects are continuous rather than discrete points.

Here we described the case where the total mass is spread uniformly over the available length or area.
However, our approach also works when this is not the case.
If the ratio of densities is bounded, the same running times hold; otherwise, this ratio will show up in the running times the same way that the longest edge length does for cases involving triangles.

We note that for points and line segments (in any dimension), the approximation scheme is free from undesired parameters, whereas for points and triangles (or simplices), the maximum edge length \(\Delta\) appears in the running time, and when neither set is a set of points, an additive term appears in the approximation.
The most interesting open question is whether either of these two artifacts can be avoided.

\bibliographystyle{plain}
\bibliography{bibliography}

\end{document}